\documentclass[10pt, conference,letterpaper]{IEEEtran}
\IEEEoverridecommandlockouts

\usepackage{cite}
\usepackage{textcomp}
\usepackage{xcolor}

\usepackage[symbol]{footmisc}
\usepackage[font=footnotesize]{caption}
\usepackage{stfloats}
\usepackage{tikz-cd}
\usepackage{pgfplots}
\usepackage {graphics}
\usepackage {graphicx}\usepackage{subfigure}
\usepackage{titlesec}
\usepackage{float}
\usepackage{indentfirst,latexsym,bm,color}
\usepackage{amsmath,amssymb,amsfonts,amsthm}
\usepackage{CJK}
\usepackage{mathrsfs}
\usepackage{multirow}
\usepackage{balance}
\usepackage{extarrows}
\usepackage{booktabs}
\usepackage{makecell}
\usepackage{caption}
\usepackage{algorithm} 
\usepackage{algpseudocode} 
\usepackage[algo2e,ruled, vlined, linesnumbered]{algorithm2e}
\usepackage{tabularx}
\usetikzlibrary{positioning,arrows.meta}
\tikzset{
     block/.style={rectangle, draw, fill=red!40, text width=6em,
                   text centered, rounded corners, minimum height=3em},
     arrow/.style={-{Stealth[]}}
}

\def\BibTeX{{\rm B\kern-.05em{\sc i\kern-.025em b}\kern-.08em
    T\kern-.1667em\lower.7ex\hbox{E}\kern-.125emX}}

\newtheorem{theorem}{Theorem}

\makeatletter
\newcommand{\algmargin}{\the\ALG@thistlm}
\makeatother
\newlength{\whilewidth}
\algdef{SE}[parWHILE]{parWhile}{EndparWhile}[1]
  {\parbox[t]{\dimexpr\linewidth-\algmargin}{%
     \hangindent\whilewidth\strut\algorithmicwhile\ #1\ \algorithmicdo\strut}}{\algorithmicend\ \algorithmicwhile}%
\algnewcommand{\parState}[1]{\State%
  \parbox[t]{\dimexpr\linewidth-\algmargin}{\strut #1\strut}}
  
\begin{document}

\title{Linear Model Against Malicious Adversaries with Local Differential Privacy\\
}

\author{Guanhong~Miao,
        A.~Adam~Ding,
        and~Samuel~S.~Wu  \textsuperscript{*}
\thanks{G. Miao and S. Wu are with University of Florida, Gainesville, FL, 32611, USA. e-mail: gmiao@ufl.edu, samwu@biostat.ufl.edu.}
\thanks{A. Ding is with Northeastern University, Boston, MA, 02115, USA. e-mail: a.ding@neu.edu.}
}


\maketitle


\begin{abstract}
Scientific collaborations benefit from collaborative learning of distributed sources, but remain difficult to achieve when data are sensitive. In recent years, privacy preserving techniques have been widely studied to analyze distributed data across different agencies while protecting sensitive information. Most existing privacy preserving techniques are designed to resist semi-honest adversaries and require intense computation to perform data analysis. Secure collaborative learning is significantly difficult with the presence of malicious adversaries who may deviates from the secure protocol. Another challenge is to maintain high computation efficiency with privacy protection. In this paper, matrix encryption is applied to encrypt data such that the secure schemes are against malicious adversaries, including chosen plaintext attack, known plaintext attack, and collusion attack. The encryption scheme also achieves local differential privacy. Moreover, cross validation is studied to prevent overfitting without additional communication cost. Empirical experiments on real-world datasets demonstrate that the proposed schemes are computationally efficient compared to existing techniques against malicious adversary and semi-honest model. 
\end{abstract}



\section{Introduction}
The demand of collaborative learning over distributed datasets increases as recent advances in computing and communication technologies. Agencies cooperate to build statistical models on aggregated datasets to obtain more accurate models. Vertical and horizontal partitioning are two common partitioning approaches to integrate distributed datasets. Vertical partitioning happens when participating agencies have datasets with different sets of features on the same sets of samples. For example, biomedical applications often need to consult records distributed among several heterogeneous domains, such as genotype data, clinical data and medical imaging, to define more accurate diagnosis for a single patient. Horizontal partitioning happens when multiple agencies have datasets with identical features for disjoint sets of samples. In many cases data are collected over different sites with the same features. For instance, hospitals in different locations have the same type of diagnosis records and other health related information for different patients. 

Privacy protection is a big challenge to perform collaborative learning as data may contain sensitive information so that data owners may not be willing to share data unless privacy is guaranteed. For instance, biomedical data integration and sharing raise public concerns that information exchange (e.g., demographics, genome sequences, medications) can put sensitive patient information at risk. A breach can have serious implications for research participants.

A variety of literatures have addressed diverse solutions for privacy preserving collaborative learning. Vaidya and Clifton \cite{vertical} developed secure protocols to find association rules over the vertically partitioned data. Nikolaenko et al. \cite{cryptographic3} proposed a secure linear regression approach for a scenario where many parties upload their data to a server to build the model. A privacy preserving linear regression protocol was investigated for vertical partitioning on high-dimensional data \cite{garble1}. Secure systems that work for both vertical and horizontal partitioning were presented in \cite{lm_01,lm_02}. Maliciously secure coopetitive learning for horizontally partitioned linear models were proposed in \cite{lm_03}.

In this paper, we develop privacy preserving schemes for linear models. Because linear models are easy to interpret and statistically robust, they are widely used in bioinformatics research \cite{application6}, financial risk analysis \cite{application7}, and are the foundation of basis pursuit techniques in signal processing. We investigate linear model schemes to achieve security against malicious adversaries (which means adversary may use any efficient attack strategy and thus may arbitrarily deviate from the protocol specification) with efficiency to permit use on relatively large datasets. Our contributions are as follows:

\begin{enumerate}
    \item Our scheme is resilient to malicious adversaries, including chosen plaintext attack, known plaintext attack, and collusion attack which compromises all but one agency. If any agency deviates from the scheme, the result is not accurate but still no sensitive information of original data is disclosed. 
     \item The proposed scheme satisfies local differential privacy, such that the probability distribution of scheme output is roughly the same for any two inputs. The output does not reveal significant information about any particular element in the input.
        \item Cross validation is feasible in the proposed schemes to prevent overfitting problem and select penalty parameters in ridge regression without additional communication cost. 
    \item The scheme has high computational efficiency to analyze large datasets with high accuracy. 
\end{enumerate}

The rest of the paper is organized as follows. Section II reviews the related work. Preliminaries are presented in Section III. In Section IV, we provide the system overview. Section V introduces the proposed scheme. Security analysis is given in Section VI. Section VII provides the performance evaluations by simulation. Finally, Section VIII concludes the paper.



\section{Related work}
Techniques for privacy preserving data analysis fall into two major categories: perturbation-based approaches and secure multiparty computation (SMC)-based approaches. Differential privacy \cite{DP} has been widely embraced by research communities as an accepted notion of privacy for statistical analysis.

Data perturbation techniques have been widely studied as a tool of privacy preserving data mining \cite{perturbation1,perturbation4,multiplicative4}. Chen et al. developed geometric perturbation \cite{perturbation4} and added noise term to enhance the security. Data utility is preserved using the geometric perturbation. The noise term drops the utility and is not ideal to build accurate models. Liu et al. proposed random projection perturbation \cite{multiplicative4} by dimension reduction approach. The dimension reduction approach loses some information of the data and large sample size is required in order to reach acceptable power. Moreover, plenty of studies focused on linear models using perturbation approaches to encrypt data. Linear regression based on matrix encryption techniques were investigated for different privacy preserving problems \cite{linearprivacy, securematrixproducts,samuel2017new,matrix_mask}. Du et al. \cite{linearprivacy} studied linear regression in Secure 2-party Computation framework where each of the two parties holds a secret data set and wants to conduct analysis on the joint data. Karr et al. \cite{securematrixproducts} used secure matrix product technique to allow multiple parties to estimate linear regression coefficients but was not immune to breaches of privacy. Wu et al. \cite{samuel2017new} investigated schemes to collect data privately granting data users access to non-sensitive personal information while sensitive information remains hidden. Matrix encryption were investigated for privacy preserving techniques in \cite{matrix_mask,private_2020,matrix_0,matrix_1,smm,mm_2019}. \cite{private_2020} investigated secure outsourcing face recognition based on elementary matrix transformation. \cite{matrix_0} studied secure algorithms for outsourcing linear equations. Secure outsourcing algorithms of matrix operations were proposed in \cite{matrix_1}. In \cite{mm_2019}, matrix filled with random integers were used for encryption by both-sided matrix multiplication which ensures robustness to known plaintext attack and brute-force attack. Sparse matrix encryption was used to design privacy preserving outsourced computation in \cite{smm}. Chen et al. \cite{matrix_mask} investigated efficient linear regression outsourcing to a cloud. The secure schemes were questioned for the vulnerability to disclosure attack and its research significance \cite{disclosure_c}. Due to the trade-off between data utility and disclosure risk, matrix encryption methods proposed in previous studies face potential disclosure risks and may release extra information of original data under certain circumstances.

Plenty of previous works utilized cryptographic techniques and SMC to control disclosure risk \cite{cryptographic1,cryptographic2,cryptographic3,garble1,lm_01,lm_02,lm_03} for secure linear models. By allowing the evaluation of arbitrary computations on encrypted data without decrypting it, homomorphic encryption (HE) schemes were predominantly applied in state-of-the-art SMC-based approaches. Hall et al. \cite{cryptographic1} proposed an iteration algorithm to compute the inversion of matrix privately for secure linear regression. Cock et al. \cite{cryptographic2} further improved the inversion protocol for the parties to compute linear regression coefficients cooperatively. Nikolaenko et al. \cite{cryptographic3} proposed a hybrid approach using garbled circuit method for a large distributed dataset among million of users. The major bottlenecks of this protocol are that the number of gates in the garbled circuit is large and the computation cost grows proportionally. Gasc\'{o}n et al. \cite{garble1} extended protocol in \cite{cryptographic3} for vertically partitioned data distributed among agencies. Conjugate gradient descent was applied to provide a more efficient computation while maintaining accuracy and convergence rate. Maliciously secure linear model, \textit{Helen}, was investigated for horizontal partitioning in \cite{lm_03}. \textit{Helen} was designed for the cases that organizations have large amount of samples (up to millions) and a smaller number of features (up to hundreds). Using homomorphic encryption and SMC protocols, \textit{Helen} is able to achieve high level of privacy protection but also requires expensive computation cost. \textit{GuardLR} \cite{lm_2022} is another secure linear regression designed to against malicious adversary. \textit{GuardLR} requires two cloud servers with one for secure training and another for secure prediction. The collusion between these two cloud servers is not allowed. 

\begin{table}[htbp]
\begin{center}
\caption{Related work of privacy preserving linear regression models. ``K-party: Yes" refers to $K(>2)$ agencies can perform the computation with equal trust (do not need to include the two non-colluding servers model).} \label{linear}
\begin{tabular}{c c c }
\toprule
Privacy scheme & K-party? & Maliciously secure? \\
\midrule
\cite{linearprivacy} & Yes & No  \\  
\cite{linearprivacy0} & Yes  & No  \\ 
\cite{securesum} & Yes &No    \\ 
\cite{securematrixproducts} & No & No   \\ 
\cite{cryptographic1} & Yes & No  \\ 
\cite{cryptographic3} &No & No  \\ 
\cite{cryptographic2} & No  &  No \\ 
\cite{garble1} & No & No  \\
\cite{lm_01} & No & No  \\
\cite{lm_02} & No & No  \\
\cite{lm_03} & Yes & Yes  \\
\cite{lm_2022} & Yes & Yes \\
Our scheme & Yes & Yes \\
\bottomrule
\end{tabular}
\end{center}
\end{table}

Table \ref{linear} summarizes main references studying privacy preserving linear model in collaborative learning setting. Apart from the two common secure categories above, a distributed computation algorithm for linear regression was given in \cite{linearprivacy0}. The limitations of this method were also introduced such as the possible disclosure risk from the coefficients. Prior secure schemes did not provide malicious security except \cite{lm_03,lm_2022} and the training process in most of them require outsourcing to two non-colluding servers. Privacy preserving ridge regression has also been investigated previously \cite{cryptographic3,garble1,lm_03}. Notably, computation burdens are bottlenecks of previous secure linear models. More specifically, approaches based on HE cryptosystems involve an encoding mechanism, i.e., scaling, that converts floating-point numbers with fixed precision to integers. Larger scaling factors yield larger encryption parameters and worse performance while smaller scaling factors yield smaller encryption parameters and better performance but outputs may vary beyond the tolerance and lead to prediction inaccuracy \cite{HE}. Moreover, SMC-based approaches expect the data owners to be online and participate in the computation throughout the entire process. 

Privacy preserving linear model achieving differential privacy (DP) has also been investigated \cite{OLS_1,OLS_2,OLS_3,OLS_4,OLS_5,OLS_6}. \cite{OLS_1} enforced DP by perturbing the objective function of the optimization problem. \cite{OLS_3} reduced the dimension of features while \cite{OLS_4} reduced the dimension of samples for DP. \cite{OLS_5} built an local differential privacy-compliant stochastic gradient descent algorithm. \cite{OLS_6} designed univariate linear regression (i.e., model only includes one feature) with DP. All these studies focused on analyzing single dataset instead of collaborative learning. Moreover, dimension reduction (either dimension of samples or features) disables cross validation or deriving model estimates for each feature. 

In this paper, we propose secure and efficient linear models for collaborative learning enabling practical implementation for high-dimensional data analysis. The proposed schemes are against malicious adversary while satisfying differential privacy.

\section{Preliminaries}

\subsection{Linear model}
The linear regression model is
\begin{displaymath}
Y=X\beta+e,~~e\sim \mathcal{N}(0,\sigma^2I)
\end{displaymath}
where $Y\in \mathbb{R}^n$ is a vector of responses, $X\in \mathbb{R}^{n\times p}$ is the feature matrix, $\beta$ is a $p\times 1$ vector of regression coefficients, $e$ is an $n\times 1$ vector of random errors and $\mathcal{N}$ denotes multivariate normal distribution. $n$ is the number of samples and $p$ is the number of features. The estimate for $\beta$ is $\hat{\beta}=(X^TX)^{-1}X^TY$. 

Ridge regression is widely used to do variable selection for high-dimensional datasets \cite{ridge}. It minimizes the residual sum of squares subject to a bound on the $L_2$-norm of the coefficients
\begin{displaymath}
\hat{\beta}_{ridge}=\underset{\beta}{argmin}\{(Y-X\beta)^T(Y-X\beta)+\lambda\beta^T\beta\}.
\end{displaymath}
Ridge solutions are given by $\hat{\beta}_{ridge}=(X^TX+\lambda I)^{-1}X^TY$ where $\lambda$ is a tuning parameter.

\subsection{Local differential privacy}

A randomized function $f$ satisfies $\epsilon$-local differential privacy if and only if for any two inputs $t$ and $t'$ in the domain of $f$, and any $s\subseteq S$ where $S$ contains $f$'s all possible output, we have
\begin{displaymath}
P(f(t)\in s)\le e^{\epsilon}P(f(t')\in s).
\end{displaymath}




\section{System overview}
 \subsection{System model}
\begin{figure}[htbp]
\centering
\includegraphics[width=3.4in]{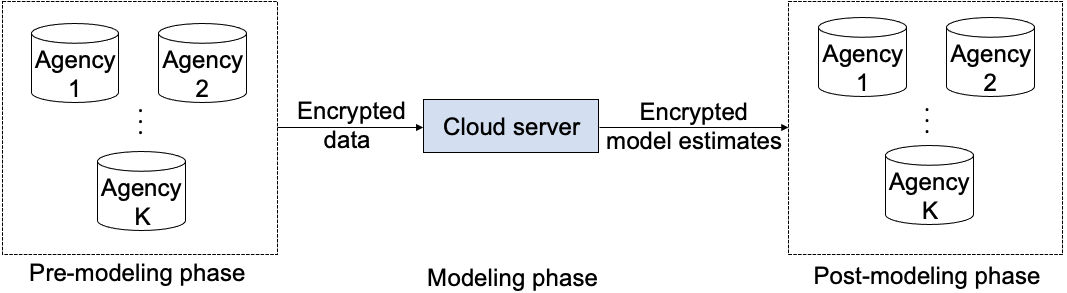}
\caption[Privacy preserving scheme framework.]{Privacy preserving scheme framework.}\label{frame}
\end{figure}
This paper focuses on collaborative learning in which data is stored by different agencies locally and they try to build linear regression model using all the data while preserving data confidentiality. Suppose there are $K$ agencies and agency $i$ has $X_i$ and $Y_i$  $(i=1,\cdots,K)$. We focus on horizontal partitioning scenario where agencies has the same set of features and different sets of samples.

The proposed privacy preserving scheme contains pre-modeling, modeling and post-modeling phase (Figure \ref{frame}). Data are encrypted in the pre-modeling phase. Encrypted data are then sent to cloud computing service provider (i.e., cloud server). The cloud server can be any agency participating in the collaborative learning. In the modeling phase, the cloud server conducts privacy preserving linear model using encrypted data. The cloud server then sends encrypted model results back to agencies. The encrypted model results are decrypted in the post-modeling phase.

\subsection{Threat model}\label{threat_model}
Assume that adversary model is malicious, i.e., agencies or the cloud server may arbitrarily deviate from the scheme specification and use any efficient attack strategy. The malicious adversary can be classified into the following categories.
\begin{itemize}
      \item The malicious agencies may generate fake data and conduct chosen plaintext attack \cite{KPA_CPA} to recover the encryption matrices generated by other agencies. 
      \item If part of the original data is disclosed, malicious adversary can recover private data by known plaintext attack \cite{multiplicative5}. 
      \item The agencies and cloud server may execute different computation than expected. 
      \item The cloud server may intentionally return a random or forged result. 
      \item A collusion attack may compromises any malicious agency or cloud server.
\end{itemize}


\textit{Out of scope attacks}: The proposed schemes do not prevent a malicious agency from inputing a bad dataset for the computation in attempt to alter model result (i.e., poisoning attack \cite{poisoning}). However the proposed scheme ensures that once an agency provides an input into the computation, the agency is bound to using the same input consistently throughout the entire computation.

 \subsection{Design goals}
The design goals are summarized as follows.
\begin{itemize}
      \item Privacy: Any agency may perform maliciously. A strong threat model includes collusions which compromise all but one agency, cloud server, and adversary outside the collaborative learning. We design encryption method resilient to malicious behavior of agencies and cloud server, including collusion attack, chosen plaintext attack, and known plaintext attack.
      \item Soundness: The privacy preserving scheme is able to verify if all the agencies and cloud server behave honestly.
      \item Efficiency: The secure scheme is computationally efficient and achieves high accuracy.
\end{itemize}






\section{The proposed scheme} \label{intro}

\subsection{Pre-modeling phase} 


We first describe the encryption method for a single agency. To simplify the notation, we use $X$ as feature matrix and $Y$ as response. Additive noise $\Delta$ is added to original dataset in the first layer, i.e, $\tilde{X}=X+\Delta$. In the second layer, data is further encrypted by row and column transformation, i.e., $X^{*}=A\tilde{X}B$ ($A$ and $B$ are randomly generated orthogonal and invertible matrix, respectively). To summarize, the encryption mechanism  $g$ is $g(X)=A(X+\Delta)B$. For response $Y$ encryption, two pseudo responses are generated with one for verification and another to enhance encryption. Define the first pseudo response $Y_{s1}=\sum_{i=1}^{p}x_i$ where $x_i$ is the $i$-th column in $X+\Delta$. The second pseudo response $Y_{s2}$ is generated randomly. Let $Y^*=A[Y,Y_{s1},Y_{s2}]C$ where $C$ is a $3\times 3$ random invertible matrix. The detailed encryption approach of each agency is given as follows. 

\subsubsection{First layer encryption}
Agency $i$ generates noise matrix $\Delta_i$ with each element following Gaussian distribution $N(0,\sigma^2)$. 
$\Delta$ is added to $X_i$ to get encrypted data $\tilde{X}_i=X_i+\Delta_i$.

\subsubsection{Second layer encryption} \label{horizontal_0}
In order to maintain data utility for linear model, $\tilde{X}_i$ $(i=1,2,\cdots,K)$ needs to be encrypted by identical $B$ in the final encryption data since the aggregated dataset is in the form of $[\tilde{X}^{T}_1,\tilde{X}^{T}_2,\cdots,\tilde{X}^{T}_K]^T$. Because each agency does not know encryption matrices generated by other agencies, we encrypt $\tilde{X}_i$ by all agencies with the commutative encryption matrix $B_i$ generated by agency $i$. This specific encryption approach guarantees that $\tilde{X}_i$ is encrypted by identical $B=\underset{i=1}{\overset{K}{\prod}}B_i$.  

Agency $i$ generates orthogonal matrix $A_{i1}$, $A_{i2}$, $\cdots$, $A_{iK}$, invertible matrix $B_{i}$ and $3\times 3$ invertible matrix $C_i$ ($i=1,2,\cdots,K$). Suppose agency $i$ has $n_i$ samples and $p$ features. To make $B_i$ commutative, each agency first generates $p\times p$ dimensional $B_0$ locally as the matrix basis using the same random seed. Then agency $i$ generates a vector of random coefficients $(b_{i1}, \cdots, b_{ip})$ and computes $B_i=\underset{j=1}{\overset{p}{\sum}}b_{ij}B_0^j$. Similarly, agency $i$ generates common invertible matrix basis $C_0$, a random vector $(c_{i1}, c_{i2}, c_{i3})$, and compute $C_i=\underset{j=1}{\overset{3}{\sum}}c_{ij}C_0^j$. $B_i$ has dimension $p\times p$, $C_i$ has dimension $3\times 3$, and the dimension of $A_{i1}$, $A_{i2}$, $\cdots$, $A_{iK}$ is $n_1\times n_1$, $n_2\times n_2$, $\cdots$, $n_K\times n_K$, respectively ($i=1,2,\cdots,K$). Agency $i$ generates pseudo responses $Y_{s1i}$ and $Y_{s2i}$. Let $X^{*}_i=A_{ii}\tilde{X}_iB_i$ and $Y^*_i=A_{ii}[Y_i,Y_{s1i},Y_{s2i}]C_i$. Agency $i$ releases $X^{*}_i$ and $Y^{*}_i$ to other agencies. Agency $j$ $(j\notin i)$ encrypts received data and releases $A_{ji}X^{*}_iB_j$ and $A_{ji}Y^{*}_iC_j$. $X^*_i$ and $Y^{*}_i$ are encrypted by all agencies in a pre-specific order.


\begin{algorithm2e}
\KwIn{$p\times p$ invertible matrix $B_0$ and $3\times 3$ invertible matrix $C_0$}
\KwOut{Encrypted feature matrix and response}
    \For{Agency $i=1,2,\ldots,K$}
    {
      generate a $p$-dimensional random vector $(b_{i1}, \cdots, b_{ip})$, a $3$-dimensional random vector $(c_{i1}, c_{i2}, c_{i3})$. noise matrix $\Delta_i$ following $N(0,\sigma^2)$ and orthogonal matrices $A_{i1},A_{i2},\ldots, A_{iK}$ with dimension $n_1\times n_1, n_2\times n_2, \ldots, n_K\times n_K$, respectively\;
      $B_i=\underset{j=1}{\overset{p}{\sum}}b_{ij}B_0^j$, $C_i=\underset{j=1}{\overset{3}{\sum}}c_{ij}C_0^j$\;
    }
   \For{Agency $i=1,2,\ldots,K$}
    {
    generate $Q_i$, a permutation of $\{1,\ldots,K\}$ with $Q_i(1)=i$\;
    generate $\tilde{X}_i=X_i+\Delta_i$, let the first pseudo response $Y_{s1i}$ be the sum of all columns in $\tilde{X}_i$\;
    generate the second pseudo response $Y_{s2i}$ randomly\;
    compute $X^{*}_i=A_{ii}\tilde{X}_iB_i$, $Y^{*}_i=A_{ii}[Y_i,Y_{s1i},Y_{s2i}]C_i$ and send to $Q_i(2)$\;
    $j=2$\;
    \While{$j\le K$}{
    Agency $Q_i(j)$ compute $X^{*}_i=A_{Q_i(j),i}X^{*}_iB_{Q_i(j)}$ and $Y^{*}_i=A_{Q_i(j),i}Y^{*}_iC_{Q_i(j)}$\;
    send $X^{*}_i$ and $Y^{*}_i$ to agency $Q_i(j+1)$\;
    $j=j+1$\;
    }
    return $X^{*}_i$ and $Y^{*}_i$\;
    }
\caption{Pre-modeling phase}\label{alg_horiz}
\end{algorithm2e}

Algorithm \ref{alg_horiz} gives detailed encryption procedures.

Suppose the communication order among agencies is $1\rightarrow 2\rightarrow \cdots \rightarrow K \rightarrow 1$, the final released dataset would be $X^*=A\tilde{X}B$ where $\tilde{X}=[\tilde{X}^T_1,\tilde{X}^T_2,\cdots,\tilde{X}^T_K]^T$, $A=\begin{scriptsize}\left( \begin{array}{cccc} A_{K1}\cdots A_{21}A_{11}&\mathbf{0}&\cdots&\mathbf{0} \\ \mathbf{0}&A_{12}A_{K2}\cdots A_{22}&\cdots&\mathbf{0} \\ \vdots & \vdots & \vdots & \vdots \\ \mathbf{0} &\mathbf{0}&\cdots&A_{(K-1)K}\cdots A_{1K}A_{KK} \end{array} \right)\end{scriptsize}$ and $B=B_1B_2\cdots B_K$ as $B_i$ $(i=1,2,\cdots,K)$ is commutative. Let $Y_a$ be the aggregated response data with 3 columns where the first column is the responses of $K$ agencies (i.e., $[Y^T_1,Y^T_2,\cdots,Y^T_K]^T$), the second column is the first pseudo response of $K$ agencies (i.e., $[Y^T_{s11},Y^T_{s12},\cdots,Y^T_{s1K}]^T$), and the third column is the second pseudo response of $K$ agencies (i.e., $[Y^T_{s21},Y^T_{s22},\cdots,Y^T_{s2K}]^T$). Then $Y^*=AY_aC$ where $C=C_1C_2\cdots C_K$.

The computation complexity of encryption matrix generation and multiplication increases when the dimension of dataset increases. For dataset with big number of samples, we partition orthogonal encryption matrix $A_{ii}$ into block diagonal matrix to improve computation efficiency. It is the same to partition $A_0$ since $A_{ii}$ is generated using matrix basis $A_0$. For example, there are $10,000$ samples in $X_i$. If $A_0$ is partitioned with block size $100$, agency $i$ generates $100$ orthogonal matrices with dimension $100\times 100$ instead of one $10,000\times 10,000$ matrix. In other words, $A_{ii}=diag(\tilde{A}_{1},\cdots,\tilde{A}_{100})$ where $\tilde{A}_i$ $(i=1,\cdots,100)$ are random orthogonal matrices. The same strategy can be applied when the dimension of features is big and agencies use partitioned block diagonal matrix $B$ as encryption matrix. 



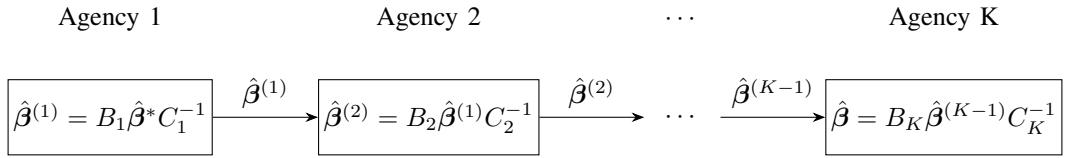
\begin{figure*}[!b]
\begin{center}
\begin{tikzpicture}
    \tikzstyle{ann} = [draw=none,fill=none,right]
    \matrix[nodes={draw,minimum size=10mm,inner sep=2pt},
        row sep=0.3cm,column sep=1.4cm] {
    \node[draw=none,fill=none] {Agency 1}; &
    \node[draw=none,fill=none] {Agency 2}; &
    \node[draw=none,fill=none] {$\cdots$}; &
    \node[draw=none,fill=none] {Agency K};\\
        \node[rectangle] (a) {$\hat{\bm{\beta}}^{(1)}=B_1\hat{\bm{\beta}}^{*}C^{-1}_1$}; &
    \node[rectangle,minimum height=10mm,minimum width=18mm] (b) {$\hat{\bm{\beta}}^{(2)}=B_{2}\hat{\bm{\beta}}^{(1)}C^{-1}_2$}; &
     \node[draw=none,fill=none] (c) {$\cdots$}; &
    \node[rectangle,minimum height=10mm,minimum width=12mm] (d) {$\hat{\bm{\beta}}=B_K\hat{\bm{\beta}}^{(K-1)}C^{-1}_K$};\\
    };
    \draw[arrow] (a)  --  (b) node [above,pos=0.5] {$\hat{\bm{\beta}}^{(1)}$};
     \draw[arrow] (b)  --  (c) node [above,pos=0.5] {$\hat{\bm{\beta}}^{(2)}$};
      \draw[arrow] (c)  --  (d) node [above,pos=0.5] {$\hat{\bm{\beta}}^{(K-1)}$};
\end{tikzpicture}
\end{center}
\caption[]{Post-modeling procedure.} \label{post_f}
\end{figure*}

\subsection{Modeling and post-modeling phase} \label{beta} 


We first analyze model results for special case $\Delta=\bm{0}$. The encrypted response is $Y^*=AY_aC$ where $Y_{a}$ contains the true response and two pseudo responses. $\hat{\bm{\beta}}^{*}=({X^{*}}^TX^{*})^{-1}{X^{*}}^TY^{*}=B^{-1}\hat{\bm{\beta}}C$. $\hat{\bm{\beta}}$ is a $p\times 3$ matrix with the first column being the true estimate, the second column corresponding to the estimate for the first pseudo response and the third column being the estimate for the second pseudo response. To compute $\hat{\bm{\beta}}$, $\hat{\bm{\beta}}^{*}$ is sent to each agency in order to eliminate $B=\underset{i=1}{\overset{K}{\prod}}B_i$ and $C=\underset{i=1}{\overset{K}{\prod}}C_i$. The order of the agency decryption can be random since $B_iB_j=B_jB_i$ and $C_iC_j=C_jC_i$. Detailed procedure of $\hat{\bm{\beta}}^{*}$ decryption is given in Figure \ref{post_f}.

\begin{figure*}[!t]
\begin{center}
\begin{tikzpicture}
    \tikzstyle{ann} = [draw=none,fill=none,right]
    \matrix[nodes={draw,minimum size=10mm,inner sep=2pt},
        row sep=0.01mm,column sep=0.2cm] {
        \node[draw=none,fill=none] {Threat model}; &
    \node[draw=none,fill=none] {Response}; &
    \node[draw=none,fill=none] {Feature matrix}; &
    \node[draw=none,fill=none] {$\hat{\beta}$}; &
    \node[draw=none,fill=none] {$\hat{\beta}$ (after post-modeling phase)};\\
     \node[draw=none,fill=none] {Non-PP}; &
        \node[draw=none,fill=none] (a) {\begin{scriptsize}$Y$\end{scriptsize}}; &
    \node[draw=none,fill=none] (b) {\begin{scriptsize}$X$\end{scriptsize}}; &
     \node[draw=none,fill=none] (c) {\begin{scriptsize}$(X^TX)^{-1}X^TY$\end{scriptsize}}; &
    \node[draw=none,fill=none] (d) {--};\\
    \node[draw=none,fill=none] {Semi-honest}; &
     \node[draw=none,fill=none] (a) {\begin{scriptsize}$AY$\end{scriptsize}}; &
    \node[draw=none,fill=none] (b) {\begin{scriptsize}$A(X+\Delta)B$\end{scriptsize}}; &
     \node[draw=none,fill=none] (c) {\begin{scriptsize}$B^{-1}((X+\Delta)^T(X+\Delta))^{-1}(X+\Delta)^TY$\end{scriptsize}}; &
    \node[draw=none,fill=none] (d) {\begin{scriptsize}$((X+\Delta)^T(X+\Delta))^{-1}(X+\Delta)^TY$\end{scriptsize}};\\
    \node[draw=none,fill=none] {Malicious adversary}; &
    \node[draw=none,fill=none] (a) {\begin{scriptsize}$AY_aC$\end{scriptsize}}; &
    \node[draw=none,fill=none] (b) {\begin{scriptsize}$A(X+\Delta)B$\end{scriptsize}}; &
     \node[draw=none,fill=none] (c) {\begin{scriptsize}$B^{-1}((X+\Delta)^T(X+\Delta))^{-1}(X+\Delta)^TY_aC$\end{scriptsize}}; &
    \node[draw=none,fill=none] (d) {\begin{scriptsize}$((X+\Delta)^T(X+\Delta))^{-1}(X+\Delta)^TY_a$\end{scriptsize}};\\
    };
   
\end{tikzpicture}
\end{center}
\caption[]{Linear regression models resilient to different threat models. Non-PP: non-privacy preserving model. $Y_a=[Y,Y_{s1},Y_{s2}]$. } \label{MA}
\end{figure*}


For $\Delta\ne\bm{0}$, $\hat{\bm{\beta}}_{\Delta}^*=B^{-1}((X+\Delta)^T(X+\Delta))^{-1}(X+\Delta)^TY_aC$. The decryption approach is the same as given above. After decryption, $\hat{\bm{\beta}}_{\Delta}=((X+\Delta)^T(X+\Delta))^{-1}(X+\Delta)^TY_a$ is a $p\times 3$ matrix with the first column being the true model estimate. 

\textit{Privacy preserving ridge regression}: In order to compute $\hat{\bm{\beta}}_{ridge}$, matrix $B^TB$ needs to be computed and released additionally from pre-modeling phase using similar procedures in Algorithm \ref{alg_horiz}. For a given $\lambda$, we have $\hat{\bm{\beta}}_{ridge}^{*}=[{X^{*}}^TX^{*}+\lambda (B^TB)^{-1}]^{-1}{X^{*}}^TY^{*}=B^{-1}\hat{\bm{\beta}}_{ridge}C$ from encrypted datasets. So $\hat{\bm{\beta}}_{ridge}=B\hat{\bm{\beta}}_{ridge}^{*}C^{-1}$. Follow the decryption procedure in Figure \ref{post_f} to get $\hat{\bm{\beta}}_{ridge}$. The pseudo response for ridge regression is different from those used in linear regression. Since different $\lambda$ yields different model estimates using previously defined pseudo response $Y_{s1}=\sum_{i=1}^{p}x_i$, let the new $Y_{s1}$ be a vector of $\bm{0}$'s for ridge regression to against malicious adversary.


\subsection{Cross validation}
Cross validation is widely used to prevent overfitting problem. It is also the golden standard to select optimal $\lambda$ for ridge regression. By partitioning orthogonal matrix $A$, the proposed scheme enables cross validation for privacy preserving linear models without additional communication cost. The procedure of partitioning $A$ has been illustrated above. For $k$-cross validation, each agency uses $k$-blocked orthogonal matrix to encrypt data and samples within each block are used as training or testing set. Our schemes are efficient and practical to change regularization parameter $\lambda$ while previous secure ridge regression models used public and fixed $\lambda$ \cite{cryptographic1,cryptographic3,garble1,lm_03}.




\section{Security analysis} \label{security}
Without loss of generality, we use $X$ to denote the original data, and $\Delta$, $A$, $B$ to denote additive noise matrix, orthogonal and invertible encryption matrices, respectively.

Matrix encryption has been widely used in previous privacy preserving studies. $X+\Delta$ has been widely used as encryption method but has disclosure risks demonstrated in previous literatures \cite{attack_1,noise_4,attack_2,attack_3}. Privacy methods with encrypted matrix in the form of $XB$ and $AX$ were studied in \cite{perturbation4,multiplicative4,linearprivacy,samuel2017new}. $XB$ and $AX$ have high data utility but often face known plaintext attack (aka, known input-output attack) and chosen plaintext attack \cite{multiplicative5,KPA_CPA}. Sparse matrix encryption in the form of $AXB$ ($A$ and $B$ are sparse matrices) were investigated with enhanced privacy in \cite{matrix_mask,private_2020,matrix_0,matrix_1,smm}. The privacy guarantee needs to be argued carefully for the sparse matrix encryption as discussed in \cite{matrix_0,matrix_1}. Matrix encrypted in the form of $AXB$ has been proved to solve different problems without releasing sensitive information in previous works \cite{private_2020,matrix_0,matrix_1,smm,mm_2019}. 


In this paper, we encrypt data from both left side and right side using dense orthogonal/invertible matrices. The framework of privacy analysis against malicious adversary is as follows.
      \begin{itemize}
      \item detection of agencies'/cloud server's malicious computation;
      \item security of encryption matrix $B$ $\rightarrow$ resilience of chosen plaintext attack;
       \item Security of $X$;
      \begin{itemize}
      \item partial prior information of $X$ disclosed $\rightarrow$ resilience of known plaintext attack; 
      \item local differential privacy;
     \end{itemize}
      \item Resilience of collusion attack.
      \end{itemize}

\subsection{Resilience to malicious adversary}
\subsubsection{Detection of agencies'/cloud server's malicious computation}
The first pseudo response $Y_{s1}$ is included in the encrypted response matrix $Y_a$ to verify the output received from the cloud (Figure \ref{MA}). Since $Y_{s1}=\sum_{i=1}^{p}x_i$ where $x_i$ is the $i$-th column in $X+\Delta$, $\hat{\beta}_{s1}$ (the second column in model estimate matrix $\hat{\bm{\beta}}$) equals a vector of 1's if each agency follows the proposed scheme. For ridge regression, $Y_{s1}=\bm{0}$ and $\hat{\beta}_{s1}=\bm{0}$. For secure collaborative learning, the following requirements need to be met to get $\hat{\beta}_{s1}=\bm{1}$ for linear regression or $\hat{\beta}_{s1}=\bm{0}$ for ridge regression. 
\begin{enumerate}
    \item For all $j\in \{1,2,\cdots,K\}$, agency $j$ computes $Y_{s1j}=\sum_{i=1}^{p}x_{ji}$ as the first pseudo response where $x_{ji}$ is the $i$-th column in $X_j+\Delta_j$.
    \item Each agency generates random orthogonal matrix, $p\times p$ invertible matrix commutative with $B_0$ and $3\times 3$ invertible matrix commutative with $C_0$ for encryption.
    \item The cloud returns model results without further perturbation.
    \item All the agencies follow protocol to decrypt received results.
\end{enumerate}
We use the privacy scheme of two-party collaborative learning as an example to illustrate violating any of the 4 requirements can result in failure to get desired $\hat{\beta}_{s1}$. Following the pre-modeling phase, the released feature matrix is 
\begin{displaymath}
\left( \begin{array}{c} A_{21}A_{11}(X_1+\Delta_1)B_1B_2 \\ A_{12}A_{22}(X_2+\Delta_2)B_2B_1 \end{array} \right)
\end{displaymath}
and the response is 
\begin{displaymath}
\left( \begin{array}{c} A_{21}A_{11}[Y_1,Y_{s11},Y_{s21}]C_1C_2 \\ A_{12}A_{22}[Y_2,Y_{s12},Y_{s22}]C_2C_1 \end{array} \right).
\end{displaymath}
If the second requirement is met, the cloud gets $\{B_2^TB_1^T[(X_1+\Delta_1)^T(X_1+\Delta_1)+(X_2+\Delta_2)^T(X_2+\Delta_2)]B_1B_2\}^{-1}B^T_2B^T_1\{[(X_1+\Delta_1)^T[Y_1,Y_{s11},Y_{s21}]C_1C_2+(X_2+\Delta_2)^T[Y_2,Y_{s12},Y_{s22}]C_1C_2\}$ from linear regression model. Suppose the cloud returns this matrix directly to each agency without perturbing it (the third requirement) and each agency decrypts it following the protocol (the fourth requirement), the decrypted model results are $[(X_1+\Delta_1)^T(X_1+\Delta_1)+(X_2+\Delta_2)^T(X_2+\Delta_2)]^{-1}\{(X_1+\Delta_1)^T[Y_1,Y_{s11},Y_{s21}]+(X_2+\Delta_2)^T[Y_2,Y_{s12},Y_{s22}]\}$.

Next we show that $Y_{s11}=\sum_{i=1}^{p}x_{1i}$ and $Y_{s12}=\sum_{i=1}^{p}x_{2i}$ ($x_{1i}$ is the $i$-th column of $X_1+\Delta_1$ and $x_{2i}$ is the $i$-th column of $X_2+\Delta_2$) guarantee that each agency gets $\hat{\beta}_{s1}=\bm{1}$. The decrypted results can be expressed in the matrix format
\begin{displaymath}
M\left( \begin{array}{c} Y_1,Y_{s11},Y_{s21} \\ Y_2,Y_{s12},Y_{s22} \end{array} \right)
\end{displaymath}
where $M\triangleq[(X^T_1+\Delta^T_1,X^T_2+\Delta^T_2)\left( \begin{array}{c} X_1+\Delta_1 \\ X_2+\Delta_2 \end{array} \right)]^{-1}(X^T_1+\Delta^T_1,X^T_2+\Delta^T_2)$.
Because 
\begin{displaymath}
\left( \begin{array}{c} 1 \\ \vdots \\ 1 \end{array} \right)=M\left( \begin{array}{c} X_1+\Delta_1 \\ X_2+\Delta_2 \end{array} \right)\left( \begin{array}{c} 1 \\ \vdots \\ 1 \end{array} \right),
\end{displaymath}
we choose 
\begin{displaymath}
\left( \begin{array}{c} Y_{s11} \\ Y_{s12} \end{array} \right)=\left( \begin{array}{c} X_1+\Delta_1 \\ X_2+\Delta_2 \end{array} \right)\left( \begin{array}{c} 1 \\ \vdots \\ 1 \end{array} \right)
\end{displaymath}
to get $\hat{\beta}_{s1}=\bm{1}$. In other words, $Y_{s11}$ and $Y_{s12}$ are the sum of columns in $X_1+\Delta_1$ and $X_2+\Delta_2$ which is the first requirement. Similarly, let 
\begin{displaymath}
\left( \begin{array}{c} Y_{s11} \\ Y_{s12} \end{array} \right)=\left( \begin{array}{c} X_1+\Delta_1 \\ X_2+\Delta_2 \end{array} \right)\left( \begin{array}{c} 0 \\ \vdots \\ 0 \end{array} \right)
\end{displaymath}
for ridge regression to get $\hat{\beta}_{s1}=\bm{0}$.

The second and the fourth requirements examine any malicious adversary performed by agencies. The third requirement examines malicious adversary performed by the cloud. So $\hat{\beta}_{s1}$ tells the truth whether any agency or the cloud deviates from the schemes and acts as malicious adversary.  

\subsubsection{Chosen plaintext attack} \label{CPA} 
Agencies participating collaborative learning may generate fake data to perform chosen plaintext attack. As an example, agency 1 performs maliciously by the following three procedures. 
\begin{enumerate}
    \item First, agency 1 sends $X^*_1=A_{11}(X_1+\Delta_1)B_1$ to agency 2. 
    \item Second, agency 2 encrypts $X^*_1$ as $X^*_{1new}=A_{21}X^*_1B_2$ where $A_{21}$ is random orthogonal matrix and $B_2=\underset{j=1}{\overset{p}{\sum}}b^{(0)}_{j}B_0^j$ ($b^{(0)}_{j}$ is random coefficient). Then agency 2 releases the encrypted data.
    \item Because only invertible encryption matrix basis $B_0$ is released to agencies, agency 1 generates random orthogonal matrix $A^+_1$ and $\hat{B}_2=\underset{j=1}{\overset{p}{\sum}}b^{*}_{j}B_0^j$ ($b^{*}_{j}$ is random parameter) and then uses equation $A_{21}X^*_1B_2=A^+_1X^*_1\hat{B}_2$ to recover $B_2$. 
\end{enumerate} 
Because $\hat{B}_2=\underset{j=1}{\overset{p}{\sum}}b^{*}_{j}B_0^j$, $\hat{B}_2$ can also be written as follows.
\begin{displaymath}
\hat{B}_2=(B_0 ~B^2_0 ~B^3_0 ~\cdots ~B^p_0)\left( \begin{array}{c} b^*_1I \\  b^*_2I \\  b^*_3I \\ \ddots  \\  b^*_pI \end{array} \right)
\end{displaymath}
where $I$ is identity matrix. So $X^*_{1new}=A^+_1X^*_1\hat{B}_2$ can be written as 
\begin{displaymath}
X^*_{1new}=A^+_1X^*_1(B_0 ~B^2_0 ~B^3_0 ~\cdots ~B^p_0)\left( \begin{array}{c} b^*_1I \\  b^*_2I \\  b^*_3I \\ \ddots  \\  b^*_pI \end{array} \right).
\end{displaymath}
Let $U\triangleq\left( \begin{array}{c} b^*_1I \\  b^*_2I \\  b^*_3I \\ \ddots  \\  b^*_pI \end{array} \right)$. The $p$ unknown parameters are all included in each column of $U$. So the above equation can be broken down into $p$ sub-equations. With the $j$-th column of $U$ being the unknown vector $(u_j)$ and the $j$-th column of $X^*_{1new}$ being $w_j$, we have 
\begin{displaymath}
w_j=A^+_1X^*_1(B_0 ~B^2_0 ~B^3_0 ~\cdots ~B^p_0)u_j, ~~~j=1,\cdots,p.
\end{displaymath}
Let $R\triangleq A^+_1X^*_1(B_0 ~B^2_0 ~B^3_0 ~\cdots ~B^p_0)$ and then $w_j=Ru_j$. The dimension of $R$ is $n\times p^2$. To solve $u_j$, we first release the restriction of $u_j$ and do not restrict to $p$ unknown parameters. The solution is related to the rank of $R$. More specifically,
\begin{enumerate}
    \item if $rank(R)<rank([R,w_j])$, there is no solution for $u_j$;
    \item if $rank(R)=rank([R,w_j])=p^2$, there is a unique solution for $u_j$;
    \item if $rank(R)=rank([R,w_j])<p^2$, there are infinite solutions for $u_j$.
\end{enumerate} 
Because \begin{scriptsize}{$rank(R)\le min\{rank(A^+_1X^*_1),rank((B_0 ~B^2_0 ~B^3_0 ~\cdots ~B^p_0))\}$}\end{scriptsize}, we have $rank(R)\le min\{n,p\}<p^2$. So it is impossible to have unique solution for $u_j$. If $rank(X^*_1)=min\{n,p\}$, we have $rank(R)=min\{n,p\}$ because $rank(B_0)=p$ and $rank(A^+_1)=n$. Then the solution has a direct relation with the dimension of $X^*_1$ as listed below.
\begin{enumerate}
    \item For $n\ge p+1$, $rank(R)=p$ and $rank([R,w_j])=p+1$. So there is no solution for $u_j$;
    \item For $n\le p$, $rank(R)=rank([R,w_j])=n<p^2$. So there are infinite solutions for $u_j$.
\end{enumerate} 
This applies for all the $p$ sub-equations ($j=1,\cdots,p$). Because orthogonal matrix for encryption is randomly generated by each agency, the true matrix $B_2$ is not a solution, i.e., $X^*_{1new}\ne A^+_1X^*_1B_2$. A toy example is given in Appendix \ref{Toy}. Each of the $p$ sub-equations derives different solutions of $b^*_j$ ($j=1,\cdots,p$) or there is no solution. Moreover, different $A^+_1$ in the equation gets different solutions. So the proposed encryption scheme is resilient to chosen plaintext attack by malicious agencies participating collaborative learning.


\subsection{Known plaintext attack} \label{kpa_0}
With both the encrypted data and partial original data released, known plaintext attack is an effective approach \cite{multiplicative5} to recover sensitive information from data encrypted by multiplicative perturbation, e.g., $AX$ and $XB$. For the proposed encryption method, we show that encryption matrices $A$ and $B$ protect against known plaintext attack.  

Suppose the adversary knows partial data (denoted as $X_{11}$) in the sensitive data. The first scenario ($\bold{I}$) is 
\[
\resizebox{1\hsize}{!}{ $X=\left( \begin{array}{c} \textcolor{pink}{X_{11}} \\ \textcolor{cyan}{X_{22}} \end{array} \right)=
\begin{bmatrix}
    \textcolor{pink}{x_{11}}      &  \textcolor{pink}{x_{12}} & \dots&  \textcolor{pink}{x_{1p}}   \\
    \vdots & \vdots & \vdots & \vdots \\
    \textcolor{pink}{x_{n_11}}      &  \textcolor{pink}{x_{n_12}} & \dots&  \textcolor{pink}{x_{n_1p}}  \\
    \textcolor{cyan}{x_{(n_1+1)1}}      &  \textcolor{cyan}{x_{(n_1+1)2}} & \dots&  \textcolor{cyan}{x_{(n_1+1)p}}  \\
    \vdots & \vdots & \vdots & \vdots \\
  \textcolor{cyan}{x_{n1}}      &  \textcolor{cyan}{x_{n2}} & \dots&  \textcolor{cyan}{x_{np}}  \\
\end{bmatrix}$}
\]
and the second scenario ($\bold{II}$) is 
\[
\resizebox{1\hsize}{!}{ $X=(\textcolor{pink}{X_{11}},\textcolor{cyan}{X_{22}})=
\begin{bmatrix}
    \textcolor{pink}{x_{11}}   & \dots   &  \textcolor{pink}{x_{1p_1}} & \textcolor{cyan}{x_{1(p_1+1)}} & \dots & \textcolor{cyan}{x_{1p}}   \\
    \textcolor{pink}{x_{21}}   & \dots   &  \textcolor{pink}{x_{2p_1}} & \textcolor{cyan}{x_{2(p_1+1)}} & \dots & \textcolor{cyan}{x_{2p}}  \\
    \vdots & \vdots & \vdots & \vdots & \vdots & \vdots \\
  \textcolor{pink}{x_{n1}}   & \dots   &  \textcolor{pink}{x_{np_1}} & \textcolor{cyan}{x_{n(p_1+1)}} & \dots &  \textcolor{cyan}{x_{np}}  \\
\end{bmatrix}$}.
\]
where the elements highlighted in pink color denote data disclosed to the adversary and elements highlighted in green color denote private data.

$\bold{I}$. For $n\times p$ dimensional $X=\left( \begin{array}{c} X_{11} \\ X_{22} \end{array} \right)$, the released data is $X^*=A(X+\Delta)=A\left( \begin{array}{c} X_{11}+\Delta_{11} \\ X_{22}+\Delta_{22} \end{array} \right)=\left( \begin{array}{c} X^*_{11} \\ X^*_{22} \end{array} \right)$. The adversary has equation $\hat{A}^TX^*=\left( \begin{array}{c} X_{11}+\hat{\Delta}_{11} \\ \hat{X}_{22}+\hat{\Delta}_{22} \end{array} \right)$ where matrices with $\hat{}$ denote recovered matrices. Because $A$ and $\hat{A}^T$ are orthogonal matrices, the adversary has equation 
\begin{footnotesize}
\begin{displaymath}
X^{*T}X^*=X^{*T}\hat{A}\hat{A}^TX^*=(X^T_{11}+\hat{\Delta}^T_{11},X^T_{22}+\hat{\Delta}^T_{22})\left( \begin{array}{c} X_{11}+\hat{\Delta}_{11} \\ \hat{X}_{22}+\hat{\Delta}_{22} \end{array} \right).
\end{displaymath}
\end{footnotesize}
The equation can be simplified as 
\begin{displaymath}
X^{*T}X^*=X^T_{11}X_{11}+\hat{X}^T_{22}\hat{X}_{22}
\end{displaymath}
by assuming $\hat{\Delta}_{11}=\hat{\Delta}_{22}=\bold{0}$. Any orthogonal transformation of $\hat{X}_{22}$ satisfying this equation can be a recovered $X_{22}$ by the adversary. Given
\begin{displaymath}
X^{*T}X^*=(X^T_{11}+\Delta^T_{11})(X_{11}+\Delta_{11})+(X^T_{22}+\Delta^T_{22})(X_{22}+\Delta_{22}),
\end{displaymath}
we have 
\begin{displaymath}
\hat{X}^T_{22}\hat{X}_{22}-X^T_{22}X_{22}=(X^T_{11}+\Delta^T_{11})(X_{11}+\Delta_{11})-X^T_{11}X_{11}+
\end{displaymath}
\begin{displaymath}
X^T_{22}\Delta_{22}+\Delta^T_{22}X_{22}+\Delta^T_{22}\Delta_{22}.
\end{displaymath}
Consider a simplified scenario assuming that the disclosed data $X_{11}$ and private data $X_{22}$ are not mixed together, i.e., they are encrypted by separate orthogonal matrices. Then we have 
\begin{displaymath}
X^{*T}X^*-X^{*T}_{11}X^*_{11}=X^{*T}X^*-X^{T}_{11}X_{11}=X^{*T}_{22}X^*_{22}.
\end{displaymath}
After getting $X^{*T}_{22}X^*_{22}$, the adversary recovers $\hat{X}_{22}$ based on $\hat{X}_{22}^T\hat{X}_{22}=X^{*T}_{22}X^*_{22}$. Since 
\begin{displaymath}
X^{*T}_{22}X^*_{22}=(X_{22}+\Delta_{22})^T(X_{22}+\Delta_{22}),
\end{displaymath}
any orthogonal transformation of $X_{22}+\Delta_{22}$ can be a possible recovered $\hat{X}_{22}$. So 
\begin{displaymath}
\hat{X}^T_{22}\hat{X}_{22}-X^T_{22}X_{22}=X^T_{22}\Delta_{22}+\Delta^T_{22}X_{22}+\Delta^T_{22}\Delta_{22}.
\end{displaymath}


$\bold{II}$. For $n\times p$ dimensional $X=(X_{11},X_{22})$, the released data is $X^*=A(X+\Delta)=A(X_{11}+\Delta_{11},X_{22}+\Delta_{22})=(X^*_{11},X^*_{22})$. Assume $X^*$ and $X_{11}$ are disclosed and the adversary tries to recover $X_{22}$ by known plaintext attack. The adversary has equation $X^*_{11}=\hat{A}(X_{11}+\hat{\Delta}_{11})$ where $\hat{A}$ and $\hat{\Delta}_{11}$ are the recovered encryption matrices by the adversary. Let $\hat{\Delta}_{11}=0$ and the simplified equation is $X^*_{11}=\hat{A}X_{11}$. Practically the adversary only knows limited information of the sensitive data $X$ (i.e., $X_{11}$ has small number of columns) and is not able to recover encryption matrix $A$. Consider the extreme case where $X_{11}$ contains at least $n$ columns and thus $X_{11}+\Delta_{11}$ is invertible. The adversary performs the following two computation steps to recover $\hat{X}_{22}$.

1. $A(X_{11}+\Delta_{11})=\hat{A}X_{11}\Rightarrow \hat{A}^TA=X_{11}(X_{11}+\Delta_{11})^{-1}$;

2. $\hat{X}_{22}=\hat{A}^TX^*_{22}=\hat{A}^TA(X_{22}+\Delta_{22})=X_{11}(X_{11}+\Delta_{11})^{-1}(X_{22}+\Delta_{22})$.




The above discussion of known plaintext attack is based on the encryption function $f(X)=A(X+\Delta)$. The recovered data derived above is further perturbed by $B$. We already prove that encryption matrix $B$ is impossible to be recovered by malicious adversary (Section \ref{CPA}). For $\Delta=\bm{0}$, we show that the proposed scheme is still secure with the encryption matrices $A$ and $B$. 

\textit{The contribution of invertible matrix $B$} Because $B$ can not be recovered and separated from $XB$, we replace $X$ with $XB$ in recovered data derived above.

\begin{itemize}
\item In scenario $\bold{I}$, the adversary has equation $X^{*T}X^*=X^T_{11}X_{11}+\hat{X}^T_{22}\hat{X}_{22}$ where $X^{*T}X^*=B^T[(X^T_{11}+\Delta^T_{11})(X_{11}+\Delta_{11})+(X^T_{22}+\Delta^T_{22})(X_{22}+\Delta_{22})]B$. So $\hat{X}^T_{22}\hat{X}_{22}-X^T_{22}X_{22}$ has lower bound as follows.
\begin{displaymath}
\hat{X}^T_{22}\hat{X}_{22}-X^T_{22}X_{22}=X^{*T}X^*-X^T_{11}X_{11}-X^T_{22}X_{22}
\end{displaymath}
\begin{displaymath}
>B^T(X^T_{22}+\Delta^T_{22})(X_{22}+\Delta_{22})B-X^T_{22}X_{22}.
\end{displaymath}
Consider a simple case without additive noise (i.e., $\Delta_{11}=\Delta_{22}=\bm{0}$). The difference between $\hat{X}^T_{22}\hat{X}_{22}$ and $X^T_{22}X_{22}$ is 
\begin{displaymath}
B^T(X^T_{11}X_{11}+X^T_{22}X_{22})B-X^T_{11}X_{11}-X^T_{22}X_{22}
\end{displaymath}
with lower bound $B^TX^T_{22}X_{22}B-X^T_{22}X_{22}$. 
\item In scenario $\bold{II}$, the columns of $X_{11}$ and $X_{22}$ are encrypted and mixed together by invertible matrix $B$. Because the adversary can not recover $B$, the adversary gets $\hat{X}_{22}=X_{11}(Z^*_{11})^{-1}Z^*_{22}$ where $(Z^*_{11},Z^*_{22})=(X_{11}+\Delta_{11},X_{22}+\Delta_{22})B$. Consider a simple case where $X_{11}$ and $X_{22}$ are encrypted by separate invertible matrices $B_1$ and $B_2$ (i.e., their columns are not mixed together). The recovered $\hat{X}_{22}$ can be simplified as $\hat{X}_{22}=X_{11}[(X_{11}+\Delta_{11})B_1]^{-1}(X_{22}+\Delta_{22})B_2$. Then $\hat{X}_{22}=X_{11}B_1^{-1}X_{11}^{-1}X_{22}B_2$ for $\Delta_{11}=\Delta_{22}=\bm{0}$. With the perturbation of the random invertible matrices $B_1$ and $B_2$, $\hat{X}_{22}$ deviates from $X_{22}$ without additive noise.
\end{itemize}

For random invertible matrix with element following normal distribution $N(0,\sigma_B)$, we show that each element in $B^TX^T_{22}X_{22}B$ has an upper bound.

The $(i,j)$-th element in $B^TX^T_{22}X_{22}B$ (i.e., element in the $i$-th row and the $j$-th column) is the product of the $i$-th row in $B^TX^T_{22}$ and the $j$-th column in $X_{22}B$. The $j$-th column in $X_{22}B$ can be expressed as 
\[
\resizebox{0.8\hsize}{!}{$\left( \begin{array}{c} x^*_{1j} \\ x^*_{2j} \\ \vdots \\ x^*_{nj} \end{array} \right)=
\begin{bmatrix}
    {x_{11}}      &  {x_{12}} & \dots& {x_{1p}}   \\
    {x_{21}}      & {x_{22}} & \dots& {x_{2p}}  \\
    \vdots & \vdots & \vdots & \vdots \\
  {x_{n1}}      &  {x_{n2}} & \dots&  {x_{np}}  \\
\end{bmatrix}\left( \begin{array}{c} b_{1j} \\ b_{2j} \\ \vdots \\b_{nj} \end{array} \right)$}
\]
where $(b_{1j},b_{2j},\cdots,b_{nj})^T$ is the $j$-th column in $B$. Because $b_{kj}$ ($k=1,\cdots,n$) follows normal distribution $N(0,\sigma_B)$, $x^*_{cj}=\underset{t=1}{\overset{p}{\sum}}x_{ct}b_{tj}$ has normal distribution $N(0,\sigma_B^2(\underset{t=1}{\overset{p}{\sum}}x_{ct}^2))$ for $c=1,\cdots,n$. Since the $i$-th row in $B^TX^T_{22}$ is the same as the $i$-th column in $X_{22}B$, the $(i,j)$-th element in $B^TX^T_{22}X_{22}B$ equals $\underset{t=1}{\overset{n}{\sum}}x^*_{ti}x^*_{tj}$. For each $t\in\{1,\cdots,n\}$, $x^*_{ti}x^*_{tj}$ follows Gamma distribution $\Gamma(1/2,2\sigma^2_B(\underset{k=1}{\overset{p}{\sum}}x_{tk}^2))$. 
Given $n$ and $X$, $\underset{k=1}{\overset{p}{\sum}}x_{tk}^2$ is fixed. For $\sigma_B\rightarrow 0$, $x^*_{ti}x^*_{tj}\rightarrow 0$. So $\underset{t=1}{\overset{n}{\sum}}x^*_{ti}x^*_{tj}\rightarrow 0$.

The encryption method $AXB$ is sufficient to protect against known plaintext attack without adding additive noise. So we set $\Delta=\bm{0}$ to ensure high data utility.

\subsection{Local differential privacy}
To achieve local differential privacy (LDP), invertible matrix $B$ is generated randomly with each element following normal distribution $N(0,\sigma_B^2)$.

As proved in \cite{DP_3,DP_2}, Johnson-Lindenstrauss (JL) transformation preserves differential privacy (DP). Given $n\times d$ matrix $X$ and $r\times n$ encryption matrix $R$ with each entry following Gaussian distribution $N(0,1)$, $RX$ preserves $(\epsilon,\delta)$-DP for a specified $r$. Here we prove $XB$ achieves LDP if entries in $B$ follow Gaussian distribution. The approaches in \cite{DP_3,DP_2} encrypt sensitive data but also need to ensure that encrypted data has acceptable data utility. In our method, the data utility is guaranteed by the property of the encryption matrix. The encryption matrix multiplication is commutative and the encryption effect can be eliminated by the post-modeling phase. Since the encryption can be eliminated in our method, the perturbation of $B$ guarantees LDP and also does not influence data utility.

\begin{theorem}\label{theorem3.} Encryption function $f(X)=XB_0$ achieves LDP if entries in encryption matrix $B_0$ follow Gaussian distribution $N(0,\sigma_{B_0}^2)$.
\end{theorem}

\begin{proof}
Consider any two samples, $x^{(1)}=(x^{(1)}_1,x^{(1)}_2,x^{(1)}_3,\cdots,x^{(1)}_p)$ and $x^{(2)}=(x^{(2)}_1,x^{(2)}_2,x^{(2)}_3,\cdots,x^{(2)}_p)$, that randomly chosen from all the possible inputs in $X$. 

\[
\resizebox{1\hsize}{!}{ $x^{(1)}B_0=(x^{(1)}_1,x^{(1)}_2,x^{(1)}_3,\cdots,x^{(1)}_p$)$
\begin{bmatrix}
    \textcolor{blue}{b_{11}}      &  \textcolor{pink}{b_{12}} & \dots&  \textcolor{cyan}{b_{1p}}   \\
    \textcolor{blue}{b_{21}}      &  \textcolor{pink}{b_{22}} & \dots&  \textcolor{cyan}{b_{2p}}  \\
    \vdots & \vdots & \vdots & \vdots \\
  \textcolor{blue}{b_{p1}}      &  \textcolor{pink}{b_{p2}} & \dots&  \textcolor{cyan}{b_{pp}}  \\
\end{bmatrix}=(\textcolor{blue}{x^{(1)}_{*1}},\textcolor{pink}{x^{(1)}_{*2}},\cdots,\textcolor{cyan}{x^{(1)}_{*p}})$,}
\]

\[
\resizebox{1\hsize}{!}{ $x^{(2)}B_0=(x^{(2)}_1,x^{(2)}_2,x^{(2)}_3,\cdots,x^{(2)}_p$)$
\begin{bmatrix}
    \textcolor{blue}{b_{11}}      &  \textcolor{pink}{b_{12}} & \dots&  \textcolor{cyan}{b_{1p}}   \\
    \textcolor{blue}{b_{21}}      &  \textcolor{pink}{b_{22}} & \dots&  \textcolor{cyan}{b_{2p}}  \\
    \vdots & \vdots & \vdots & \vdots \\
  \textcolor{blue}{b_{p1}}      &  \textcolor{pink}{b_{p2}} & \dots&  \textcolor{cyan}{b_{pp}}  \\
\end{bmatrix}=(\textcolor{blue}{x^{(2)}_{*1}},\textcolor{pink}{x^{(2)}_{*2}},\cdots,\textcolor{cyan}{x^{(2)}_{*p}})$.}
\]

\begin{figure}[h]
\includegraphics[width=3.4in]{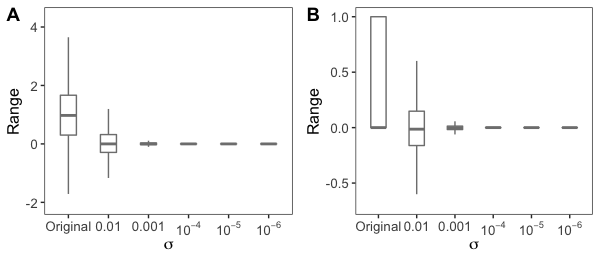}
\caption[]{For any row in $X$ (i.e., $x$), each element of $xB_0$ $\rightarrow 0$ as $\sigma_{B_0}\rightarrow 0$. The first boxplot shows the range of original $x$ and other boxplots show the range of encrypted data (i.e., $xB_0$) given different $\sigma_{B_0}$ for encryption matrix $B_0$ generation. $\bold{A}$ (continuous scenario): each element of $x$ follows $N(1,1)$; $\bold{B}$ (binary scenario): each element of $x$ is 1 or 0 with the probability of 1/2. }\label{DP_1}
\end{figure}

\begin{figure}[h]
\includegraphics[width=3.4in]{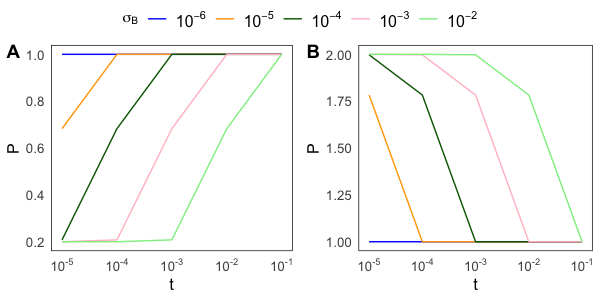}
\caption[]{$\sigma_{B_0}\rightarrow 0$ leads to $P(x^{(1)}_{*1}\in (-t,t))/P(x^{(2)}_{*1}\in (-t,t))\rightarrow 1$ with small $\sigma_{B_0}$ for small $t$ (assume $||x^{(1)}||_2=1$). ``P" (y axis) in the plot denotes $P(x^{(1)}_{*1}\in (-t,t))/P(x^{(2)}_{*1}\in (-t,t))$. A: $||x^{(2)}||_2=5$; B: $||x^{(2)}||_2=0.5$. }\label{DP}
\end{figure}

Each element in $x^{(1)}B_0$ or $x^{(2)}B_0$ is the linear combination of $B_0$ as shown in the matrix multiplication above. Because each element in $B_0$ follows normal distribution $N(0,\sigma_{B_0}^2)$, each element in $x^{(1)}B_0$ also follows normal distribution. Figure \ref{DP_1} shows that the values of $xB_0$ is close to 0 as $\sigma_{B_0}\rightarrow 0$ for any row $x$ in $X$. Specifically, each element in $x^{(1)}B_0$ follows $N(0,||x^{(1)}||^2_2\sigma_{B_0}^2)$ and each element in $x^{(2)}B_0$ follows $N(0,||x^{(2)}||^2_2\sigma_{B_0}^2)$. When $\sigma_{B_0}\rightarrow0$, these two distributions are close to each other. We use the first element $x^{(1)}_{*1}$ in $x^{(1)}B_0$ and the first element $x^{(2)}_{*1}$ in $x^{(2)}B_0$ as an example to show LDP achieved by encryption matrix $B_0$.
\begin{displaymath}
\frac{P(x^{(1)}_{*1}\in (-t,t))}{P(x^{(2)}_{*1}\in (-t,t))}=\frac{erf(t/(\sqrt{2}||x^{(1)}||_2\sigma_{B_0}))}{erf(t/(\sqrt{2}||x^{(2)}||_2\sigma_{B_0}))}
\end{displaymath}
where $erf$ is Gauss error function. For any given $||x^{(1)}||_2$, $||x^{(2)}||_2$ and $t$, there exists a $\sigma_{B_0}\rightarrow 0$ such that 
\begin{displaymath}
\frac{P(x^{(1)}_{*1}\in (-t,t))}{P(x^{(2)}_{*1}\in (-t,t))}\rightarrow 1. 
\end{displaymath}
Figure \ref{DP} shows $\frac{P(x^{(1)}_{*1}\in (-t,t))}{P(x^{(2)}_{*1}\in (-t,t))}\rightarrow 1$ when $\sigma_{B_0}\rightarrow 0$ for two examples.
In other words, the encryption method achieves local differential privacy. 
\end{proof}


\begin{theorem}\label{theorem4.} Given the commutative encryption matrix $B=\underset{j=1}{\overset{p}{\sum}}b_{j}B_0^j$, $f(X)=XB$ achieves LDP.
\end{theorem}

\begin{proof}
Suppose data is encrypted in the form of $X(\underset{j=1}{\overset{p}{\sum}}b_{j}B_0^j)$. The encrypted data can be expressed as $b_1XB_0+b_2XB_0^2+\cdots+b_pXB_0^p=XB_0(b_1+b_2B_0+\cdots+b_pB_0^{p-1})$. It is encrypted by two encryption functions $f_1(X)=XB_0$ and $f_2(X^{*})=X^{*}(b_1+b_2B_0+\cdots+b_pB_0^{p-1})$. $f_1(X)=XB_0$ achieves LDP if entries in $B_0$ follow Gaussian distribution (Theorem \ref{theorem3.}). According to closure under postprocessing property \cite{DP_1}, $f(X)=f_2(f_1(X))$ achieves local differential privacy.
\end{proof}

Similarly, $AXB$ achieves LDP based on closure under postprocessing property. With the encryption matrix $B$, the cloud server gets encrypted model estimate $B^{-1}\hat{\bm{\beta}}C$ (Section \ref{beta}) where $B$ and $C$ are encryption matrices. The invertible encryption matrices are decrypted in the post-modeling phase to get accurate model estimate and prediction accuracy remains the same as non-secure model. Data $X$ remains encrypted and secure all the time.

\subsection{Collusion attack} \label{collusion} 
The proposed encryption scheme is resilient to malicious adversary which compromises all but one agency. In this study, any released data is accessible to all agencies and the proposed encryption scheme is resilient to chosen plaintext attack and known plaintext attack. The resilience of chosen plaintext attack shows that encryption matrix $B$ can not be recovered. Any collusion among agencies is not able to provide extra information for the attack. Because the data is encrypted before release and the encryption matrix can not be recovered, the agency collusion does not increase disclosure risk. Moreover, the cloud and the outside adversary knows less prior information of the data and the encryption scheme, their attack abilities are weaker than agencies who participate collaborative learning.

\section{Performance evaluation} 

In the proposed schemes, data encryption and decryption in the pre-modeling and post-modeling phase contribute to the increasing cost while the modeling phase has the same cost as non-private model computation.

We perform experiments using four datasets from the UCI repository \cite{UCI}. All the experiments are performed in Matlab on University of Florida Hipergator 3.0 with 1 CPU and 4 RAMs.   

\textit{YearPredictionMSD}:  This dataset contains 515,345 songs with 90 features. The goal is to predict a song's published year using these 90 features. The prediction of linear regression model is rounded as the predicted year. 

\textit{Thyroid disease dataset}: Thyroid contains 7200 samples and 21 features. The response is normal (not hypothyroid) or abnormal. 

\textit{Diabetes 130-US hospitals}:  This dataset was developed to identify factors related to readmission for patients with diabetes. The features were preprocessed following the procedure in \cite{diabetes}. A total of 69,977 samples and 42 features were used to build model predicting if the patient was readmitted within 30 days of discharge (yes/no). 

\textit{Default of credit card clients Data Set}:  There are 30,000 samples and 23 features in the Taiwan credit dataset. The goal is to predict default payment (yes/no).

We assume the number of samples in each dataset is equally split into $K$ agencies while each subset contains all the features. The computation cost of the pre-modeling phase is increased when the number of agencies $K$ is increased. The privacy level and data utility are not related to $K$ since the total number of samples is fixed.

\begin{figure}[h]
\includegraphics[width=3.4in]{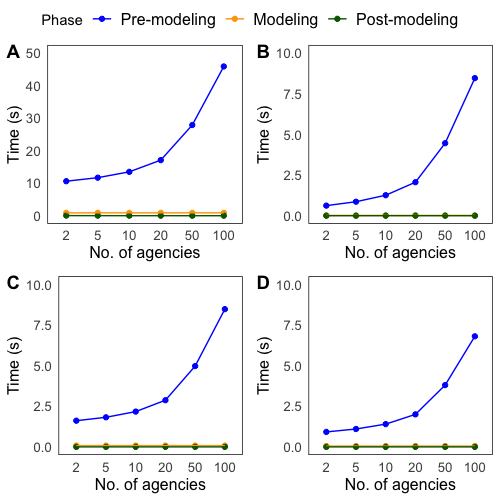}
\caption[]{Computation time of privacy preserving linear model for four datasets. A: YearPredictionMSD; B: Thyroid disease data; C: Diabetes data; D: Credit card data. }\label{time}
\end{figure}

\subsection{Efficiency}

The experiments on real data show that our privacy preserving scheme is efficient (Figure \ref{time}) and requires less computation time compared to existing secure schemes. It takes 46 seconds to analyze \textit{YearPredictionMSD} data if 100 agencies participate in the collaborative learning and each agency has around 5,000 samples. A maliciously secure linear model, \textit{Helen}, takes around 1.5 hours to finish model computation for four collaborating agencies with each holding 1,000 samples \cite{lm_03}. Another maliciously secure linear regression scheme, \textit{GuardLR} \cite{lm_2022}, takes around 218 seconds to build two-party secure linear model using \textit{Thyroid disease} data. It takes 30 seconds to perform prediction on \textit{Diabetes} data in a differentially private scheme \cite{logistic_DP}. \textit{Credit card} data was used to evaluate the secure prediction scheme proposed in \cite{logistic_2018}. However, the model takes more than 1 hour assuming the entire sample is partitioned into 16 disjoint subsets. For \textit{Thyroid disease} data, \textit{Diabetes} data, and \textit{Credit card} data, our scheme takes less than 10 seconds to conduct secure linear regression. Our privacy preserving scheme is in general more efficient than other secure schemes proposed previously.

\subsection{Privacy protection}
We choose the first 100 samples in each dataset and encrypts data with orthogonal matrix $A$ and invertible matrix $B$ with zero mean and SD=0.001. We show that the proposed encryption scheme is resilient to known plaintext attack (scenario $\bold{I}$ and $\bold{II}$).

Because the number of features (i.e., $p$) is small in each dataset, it is possible for the adversary to collect all the $p$ features of some samples and then perform known plaintext attack under scenario $\bold{I}$. Suppose $X_{11}$ contains the first $p$ samples and $p$ features. $X_{22}$ contains the remaining $100-p$ samples with $p$ featues. Figure \ref{heatmap} shows that the recovered matrix $\hat{X}_{22}^T\hat{X}_{22}$ shrinks to 0 as $\sigma_B\rightarrow 0$ for the simplified scenario $\bold{I}$ discussed in Section \ref{kpa_0}. Because the encryption matrix $B$ can be recovered in the post-modeling phase, we choose $\sigma_B\rightarrow 0$ to protect against known plaintext attack and also to achieve local differential privacy. 

\begin{figure}[!t]
\includegraphics[width=3.6in]{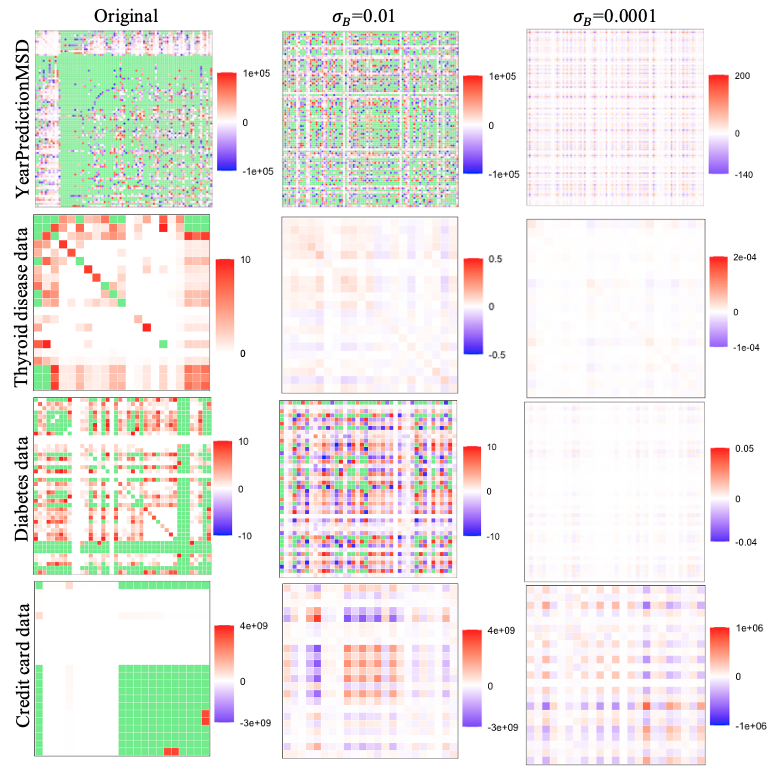}
\caption[]{Heatmap showing the differences of the recovered matrix ($\hat{X}_{22}^T\hat{X}_{22}$) for different $\sigma_B$. Blue color denotes value $<$ 0 and red color denotes value $>$ 0. Green color denotes values beyond the ranges shown in the color bar. }\label{heatmap}
\end{figure}

To simulate known plaintext attack under scenario $\bold{II}$, we assume the first $p-5$ features are disclosed, i.e., $X_{11}$ contains $p-5$ features with $p-5$ samples. $X_{22}$ contains the other 5 features for these $p-5$ samples. Figure \ref{II} shows the difference between recovered data $\hat{X}_{22}$ and original data $X_{22}$ for a randomly selected feature among 5 non-disclosed features. Different from scenario $\bold{I}$, $\sigma_B$ does not affect the amplitude of the recovered data. Because the invertible matrix is randomly generated, the perturbation is sufficient and the encrypted data is resilient to the attack under scenario $\bold{II}$.

\begin{figure}[!t]
\includegraphics[width=3.3in]{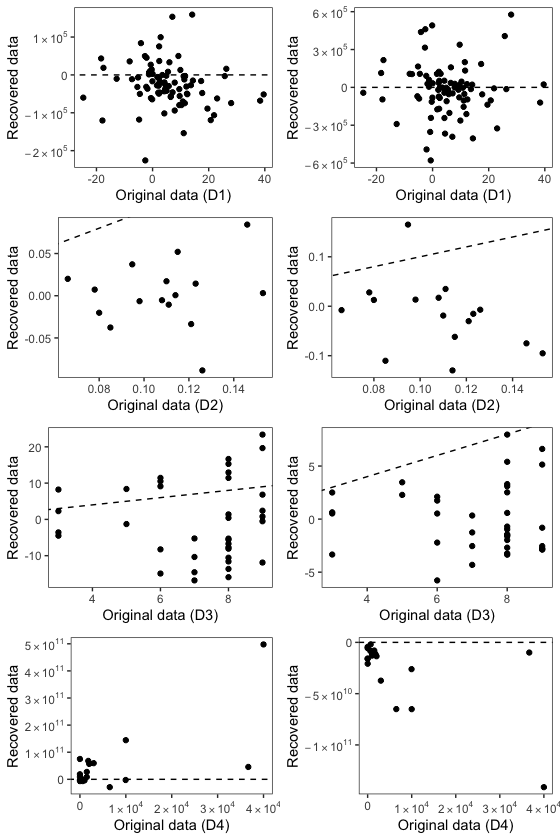}
\caption[]{Deviation of recovered data from original data given $\sigma_B=0.01$ (left plots) and $\sigma_B=0.0001$ (right plots). The dashed line denote $y=x$. D1: YearPredictionMSD; D2: Thyroid disease data; D3: Diabetes data; D4: Credit card data. }\label{II}
\end{figure}



\subsection{Prediction accuracy} 
We evaluate model performance using 10-fold cross validation with 10 iterations. In the cross validation, samples are randomly split into 10 folds with the equal size with 9 folds for training and 1 fold for testing. We use mean square error (MSE) to evaluate the prediction performance for continuous response and the area under the receiver operating characteristic curve (AUC) to evaluate accuracy in classification problems. The average MSE and AUC of 10 iterations for four datasets are shown in Figure \ref{auc}. The variation of MSE/AUC across different $\sigma_{B}$ is caused by different sample selection in the cross validation. Because model estimate is decrypted in the post-modeling phase, our scheme (regardless of the choice of $\sigma_{B}$) gets the same model result compared to non-privacy preserving linear regression using same samples for training and testing. 

\begin{figure}[h]
\includegraphics[width=3.5in]{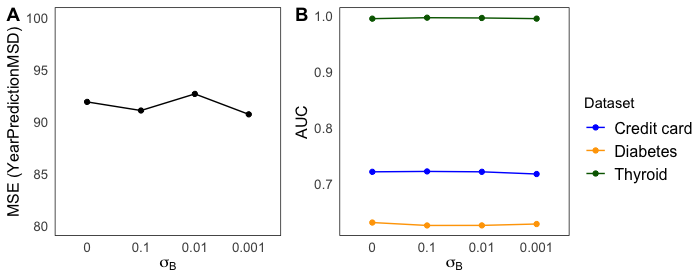}
\caption[]{Model accuracy not affected by the change of $\sigma_B$. }\label{auc}
\end{figure}

The proposed privacy preserving linear regression is efficient and has competitive accuracy when analyzing binary response in some cases. As discussed in previous studies, it is not restricted to use logistic regression with a binary response and there are compelling substantive arguments for preferring linear regression in many cases \cite{LR_1, LR_2}. Compared to existing differentially private logistic regression \cite{logistic_DP,logistic_2018}, our privacy preserving linear regression achieves similar accuracy levels and requires less computation cost. Moreover, we provide higher level of privacy protection which is resilient to malicious adversary.

\section{Conclusion}
In this paper, we propose efficient privacy preserving schemes for collaborative linear model when data is distributed among different agencies. The proposed scheme is against malicious adversaries, including chosen plaintext attack, known plaintext attack, and collusion attack. The proposed scheme also satisfies local differential privacy. It is efficient to conduct cross validation without additional communication cost. The experimental analysis shows that our scheme is more efficient than existing secure linear regression techniques against malicious adversary. Our privacy preserving scheme requires data encrypted by each agency participated in the collaborative learning which increases communication cost. As the recent development of data compression technology \cite{compression}, data can be compressed before transforming among agencies to reduce communication cost. In the future, we are interested in extending the privacy preserving schemes to other statistical models. 







\appendices




\section{Toy example of chosen plaintext attack} \label{Toy}
We give an example to show that the proposed encryption scheme is resilient to chosen plaintext attack. We set $n=p=3$ and follow the procedure described in Section \ref{CPA}. 
\begin{enumerate}
    \item Generate a random invertible matrix $B_0$ with each element following normal distribution $N(0,1)$.
    \item Agency 1 generates a $3\times 3$ random matrix $X^*_1$ with each element following normal distribution $N(1,1)$ and sends $X^*_1$ to agency 2. 
    \item Agency 2 encrypts $X^*_1$ as $X^*_{1new}=A_2X^*_1B_2$ where $A_2$ is random orthogonal matrix and $B_2=8B_0+0.3B^2_0-2B^3_0$.
    \item Agency 1 sets identity matrix as the orthogonal encryption matrix (i.e., $A^+_1=I$) and $\hat{B}_2=b^*_1B_0+b^*_2B^2_0+b^*_3B^3_0$. Equation $X^*_{1new}=A^+_1X^*_1\hat{B}_2=X^*_1\hat{B}_2$ is used to recover $B_2$. 
\end{enumerate} 
\textit{Example} We first list the matrices generated for the chosen plaintext attack. \begin{scriptsize}{$B_0=\left( \begin{array}{ccc} -0.626 & 1.595 & 0.487 \\ 0.184 & 0.330 & 0.738 \\ -0.836 & -0.820 & 0.576 \end{array} \right)$, $X^*_1=\left( \begin{array}{ccc} 0.695 & 0.379 & 0.955 \\ 2.512 & -1.215 & 0.984 \\ 1.390 & 2.125 & 1.944 \end{array} \right)$, $X^*_{1new}=\left( \begin{array}{ccc} 7.517 & -5.452 & -6.865 \\ 11.13 & -16.98 & -2.897 \\ 17.12 & -23.77 & -38.04 \end{array} \right)$. }   $X^*_{1new}=X^*_1\hat{B}_2=X^*_1(B_0~B^2_0~B^3_0)\left( \begin{array}{c} b^*_1I \\ b^*_2I \\ b^*_3I \end{array} \right)=X^*_1(B_0~B^2_0~B^3_0)\left( \begin{array}{ccc} b^*_1 & 0 & 0 \\ 0 & b^*_1 & 0 \\ 0 & 0 & b^*_1 \\ b^*_2 & 0 & 0 \\ 0 & b^*_2 & 0 \\ 0 & 0 & b^*_2 \\ b^*_3 & 0 & 0 \\ 0 & b^*_3 & 0 \\ 0 & 0 & b^*_3  \end{array} \right)$. \end{scriptsize} It can be broken down into 3 sub-equations. Let $R=X^*_1(B_0~B^2_0~B^3_0)$. The first equation is $w_1=Ru_1$ where $w_1$ is the first column of $X^*_{1new}$ and $u_1$ is the first column of \begin{scriptsize}$\left( \begin{array}{c} b^*_1I \\ b^*_2I \\ b^*_3I \end{array} \right)$\end{scriptsize}. The reduced row echelon form of $(R,w_1)$ is \begin{tiny}{$\left( \begin{array}{cccccccccc} 1 & 0 & 0 & -0.63 & 1.60 & 0.49 & 0.28 & -0.87 & 1.15 & -3.63 \\ 0 & 1 & 0 & 0.18 & 0.33 & 0.74 & -0.67 & -0.20 & 0.76 & -0.36 \\ 0 & 0 & 1 & -0.84 & -0.82 & 0.58 & -0.11 & -2.08 & -0.68 & 2.96 \end{array} \right)$}\end{tiny}. So the sub-equation has infinite solutions $u_1=$ \begin{scriptsize} $\left( \begin{array}{c} -0.63 \\ 0.18 \\ -0.84 \\ 1 \\ 0 \\ 0 \\ 0 \\ 0 \\ 0 \end{array} \right)v_4+\left( \begin{array}{c} 1.60 \\ 0.33 \\ -0.82 \\ 0 \\ 1 \\ 0 \\ 0 \\ 0 \\ 0 \end{array} \right)v_5+\left( \begin{array}{c} 0.49 \\ 0.74 \\ 0.58 \\ 0 \\ 0 \\ 1 \\ 0 \\ 0 \\ 0 \end{array} \right)v_6+\left( \begin{array}{c} 0.28 \\ -0.67 \\ -0.11 \\ 0 \\ 0 \\ 0 \\ 1 \\ 0 \\ 0 \end{array} \right)v_7+\left( \begin{array}{c} -0.87 \\ -0.20 \\ -2.08 \\ 0 \\ 0 \\ 0 \\ 0 \\ 1 \\ 0 \end{array} \right)v_8+\left( \begin{array}{c} 1.15 \\ 0.76 \\ -0.68 \\ 0 \\ 0 \\ 0 \\ 0 \\ 0 \\ 1 \end{array} \right)v_9+\left( \begin{array}{c} -3.63 \\ -0.36 \\ 2.96 \\ 0 \\ 0 \\ 0 \\ 0 \\ 0 \\ 0 \end{array} \right)$\end{scriptsize} where $v_i$ ($i=4,\cdots,9$) can be any value. Then we add the restriction back to $u_1$ (i.e., the second, third, 5th, 6th, 8th, 9th elements equal 0). So $v_5=v_6=v_8=v_9=0$, $v_4=b^*_2$, $v_7=b^*_3$, \begin{scriptsize} $\left( \begin{array}{c} -0.63 \\ 0.18 \\ -0.84 \end{array} \right)v_4+\left( \begin{array}{c} 0.28 \\ -0.67 \\ -0.11 \end{array} \right)v_7+\left( \begin{array}{c} -3.63 \\ -0.36 \\ 2.96 \end{array} \right)=\left( \begin{array}{c} b^*_1 \\ 0 \\ 0 \end{array} \right)$ \end{scriptsize}. Agency 1 gets $(b_1,b_2,b_3)=(-5.71,3.47,0.40)$. Similarly, the second sub-equation gets $(b_1,b_2,b_3)=(-11.0,-1.65,4.15)$ and the third sub-equation gets $(b_1,b_2,b_3)=(-3.45,-4.90,5.24)$.

The example shows that different sub-equations get different solutions for unknown parameters of $B_2$ and none of the solutions are close to the true values $(b_1,b_2,b_3)=(8,0.3,-2)$.

\bibliographystyle{IEEEtran}
\bibliography{references}

\begin{thebibliography}{10}
\providecommand{\url}[1]{#1}
\csname url@samestyle\endcsname
\providecommand{\newblock}{\relax}
\providecommand{\bibinfo}[2]{#2}
\providecommand{\BIBentrySTDinterwordspacing}{\spaceskip=0pt\relax}
\providecommand{\BIBentryALTinterwordstretchfactor}{4}
\providecommand{\BIBentryALTinterwordspacing}{\spaceskip=\fontdimen2\font plus
\BIBentryALTinterwordstretchfactor\fontdimen3\font minus
  \fontdimen4\font\relax}
\providecommand{\BIBforeignlanguage}[2]{{%
\expandafter\ifx\csname l@#1\endcsname\relax
\typeout{** WARNING: IEEEtran.bst: No hyphenation pattern has been}%
\typeout{** loaded for the language `#1'. Using the pattern for}%
\typeout{** the default language instead.}%
\else
\language=\csname l@#1\endcsname
\fi
#2}}
\providecommand{\BIBdecl}{\relax}
\BIBdecl

\bibitem{vertical}
J.~Vaidya and C.~Clifton, ``Privacy preserving association rule mining in
  vertically partitioned data,'' in \emph{Proceedings of the Eighth ACM SIGKDD
  International Conference on Knowledge Discovery and Data Mining}, ser. KDD
  '02, New York, NY, USA, 2002, p. 639–644.

\bibitem{cryptographic3}
V.~Nikolaenko, U.~Weinsberg, S.~Ioannidis, M.~Joye, D.~Boneh, and N.~Taft,
  ``Privacy-preserving ridge regression on hundreds of millions of records.''
  in \emph{In 2013 IEEE Symposium on Security and Privacy}, 2013, pp. 334--348.

\bibitem{garble1}
A.~Gasc\'{o}n, P.~Schoppmann, B.~Balle, M.~Raykova, J.~Doerner, S.~Zahur, and
  D.~Evans, ``Privacy-preserving distributed linear regression on
  high-dimensional data.'' \emph{Proceedings on Privacy Enhancing
  Technologies}, vol.~4, pp. 345--364, 2017.

\bibitem{lm_01}
I.~Giacomelli, S.~Jha, M.~Joye, C.~Page, and K.~Yoon, ``Privacy-preserving
  ridge regression with only linearly-homomorphic encryption,'' \emph{IACR
  Cryptology ePrint Archive}, pp. 243--261, 06 2018.

\bibitem{lm_02}
P.~{Mohassel} and Y.~{Zhang}, ``Secure{ML}: A system for scalable
  privacy-preserving machine learning,'' in \emph{2017 IEEE Symposium on
  Security and Privacy (SP)}, 2017, pp. 19--38.

\bibitem{lm_03}
W.~{Zheng}, R.~A. {Popa}, J.~E. {Gonzalez}, and I.~{Stoica}, ``Helen:
  Maliciously secure coopetitive learning for linear models,'' in \emph{2019
  IEEE Symposium on Security and Privacy (SP)}, 2019, pp. 724--738.

\bibitem{application6}
Y.~Hu, A.~Shmygelska, D.~Tran, N.~Eriksson, J.~Tung, and D.~Hinds, ``{GWAS} of
  89,283 individuals identifies genetic variants associated with self-reporting
  of being a morning person,'' \emph{Nature Communications}, vol.~7, p. 10448,
  2016.

\bibitem{application7}
K.~Valaskova, T.~Kliestik, L.~Svabova, and P.~Adamko, ``Financial risk
  measurement and prediction modelling for sustainable development of business
  entities using regression analysis,'' \emph{Sustainability}, vol.~10, no.~7,
  2018.

\bibitem{DP}
C.~Dwork and A.~Roth, ``The algorithmic foundations of differential privacy,''
  \emph{Found. Trends Theor. Comput. Sci.}, vol.~9, no. 3–4, p. 211–407,
  Aug. 2014.

\bibitem{perturbation1}
C.~Aggarwal and P.~Yu, ``A condensation approach to privacy preserving data
  mining.'' in \emph{Proceedings of International Conference on Extending
  Database Technology}, Heraklion,Crete,Greece, 2004.

\bibitem{perturbation4}
K.~Chen and L.~Liu, ``Geometric data perturbation for privacy preserving
  outsourced data mining.'' \emph{Knowledge and Information Systems}, vol.~29,
  no.~3, pp. 657--695, 2011.

\bibitem{multiplicative4}
K.~Liu, H.~Kargupta, and J.~Ryan, ``Random projection-based multiplicative data
  perturbation for privacy preserving distributed data mining.'' \emph{IEEE
  Transactions on Knowledge and Data Engineering}, vol.~18, no.~1, pp. 92--106,
  2006.

\bibitem{linearprivacy}
W.~Du, Y.~S. Han, and S.~Chen, ``Privacy-preserving multivariate statistical
  analysis: linear regression and classification.'' in \emph{Proceedings of the
  4th SIAM International Conference on Data Mining}, Lake Buena Vista, Florida,
  USA, April 2004.

\bibitem{securematrixproducts}
A.~Karr, X.~Lin, A.~Sanil, and J.~Reiter, ``Privacy-preserving analysis of
  vertically partitioned data using secure matrix products.'' \emph{Journal of
  Official Statistics}, vol.~25, no.~1, pp. 125--138, 2009.

\bibitem{samuel2017new}
S.~Wu, S.~Chen, D.~Burr, and L.~Zhang, ``A new data collection technique for
  preserving privacy.'' \emph{Journal of Privacy and Confidentiality}, vol.~7,
  no.~3, p.~5, 2017b.

\bibitem{matrix_mask}
F.~{Chen}, T.~{Xiang}, X.~{Lei}, and J.~{Chen}, ``Highly efficient linear
  regression outsourcing to a cloud,'' \emph{IEEE Transactions on Cloud
  Computing}, vol.~2, no.~4, pp. 499--508, 2014.

\bibitem{private_2020}
Y.~{Zhang}, X.~{Xiao}, L.~{Yang}, Y.~{Xiang}, and S.~{Zhong}, ``Secure and
  efficient outsourcing of {PCA}-based face recognition,'' \emph{IEEE
  Transactions on Information Forensics and Security}, vol.~15, pp. 1683--1695,
  2020.

\bibitem{matrix_0}
X.~Chen, X.~Huang, J.~Li, J.~Ma, W.~Lou, and D.~S. Wong, ``New algorithms for
  secure outsourcing of large-scale systems of linear equations,'' \emph{IEEE
  Transactions on Information Forensics and Security}, vol.~10, no.~1, pp.
  69--78, 2015.

\bibitem{matrix_1}
S.~Zhang, C.~Tian, H.~Zhang, J.~Yu, and F.~Li, ``Practical and secure
  outsourcing algorithms of matrix operations based on a novel matrix
  encryption method,'' \emph{IEEE Access}, vol.~7, pp. 53\,823--53\,838, 2019.

\bibitem{smm}
L.~Zhao and L.~Chen, ``Sparse matrix masking-based non-interactive verifiable
  (outsourced) computation, revisited,'' \emph{IEEE Transactions on Dependable
  and Secure Computing}, vol.~17, no.~6, pp. 1188--1206, 2020.

\bibitem{mm_2019}
M.~Dzwonkowski and R.~Rykaczewski, ``Secure quaternion feistel cipher for dicom
  images,'' \emph{IEEE Transactions on Image Processing}, vol.~28, no.~1, pp.
  371--380, 2019.

\bibitem{disclosure_c}
Z.~{Cao}, L.~{Liu}, and O.~{Markowitch}, ``Comment on “highly efficient
  linear regression outsourcing to a cloud”,'' \emph{IEEE Transactions on
  Cloud Computing}, vol.~7, no.~3, pp. 893--893, 2019.

\bibitem{cryptographic1}
R.~Hall, S.~Fienberg, and Y.~Nardi, ``Secure multiple linear regression based
  on homomorphic encryption.'' \emph{Journal of Official Statistics}, vol.~27,
  no.~4, pp. 669--691, 2011.

\bibitem{cryptographic2}
M.~Cock, R.~Dowsley, A.~Nascimento, and S.~Newman, ``Fast, privacy preserving
  linear regression over distributed datasets based on pre-distributed data.''
  in \emph{Proceedings of the 8th ACM Workshop on Artificial Intelligence and
  Security}, Denver,Colorado,USA, October 2015.

\bibitem{lm_2022}
Z.~Ma, J.~Ma, Y.~Miao, X.~Liu, K.-K.~R. Choo, Y.~Gao, and R.~H. Deng,
  ``Verifiable data mining against malicious adversaries in industrial internet
  of things,'' \emph{IEEE Transactions on Industrial Informatics}, vol.~18,
  no.~2, pp. 953--964, 2022.

\bibitem{linearprivacy0}
A.~Sanil, A.~Karr, X.~Lin, and J.~Reiter, ``Privacy preserving regression
  modelling via distributed computation.'' in \emph{10th ACM SIGKDD
  International Conference on Knowledge Discovery and Data Mining (KDD)},
  Seattle, WA, USA, August 2004.

\bibitem{securesum}
A.~Karr, X.~Lin, A.~Sanil, and J.~Reiter, ``Secure regression on distributed
  databases.'' \emph{Journal of Computational and Graphical Statistics},
  vol.~14, no.~2, pp. 263--279, 2005.

\bibitem{HE}
R.~Dathathri, O.~Saarikivi, H.~Chen, K.~Laine, K.~Lauter, S.~Maleki,
  M.~Musuvathi, and T.~Mytkowicz, ``{CHET}: An optimizing compiler for
  fully-homomorphic neural-network inferencing,'' in \emph{Proceedings of the
  40th ACM SIGPLAN Conference on Programming Language Design and
  Implementation}, ser. PLDI 2019.\hskip 1em plus 0.5em minus 0.4em\relax New
  York, NY, USA: Association for Computing Machinery, 2019, p. 142–156.

\bibitem{OLS_1}
J.~Zhang, Z.~Zhang, X.~Xiao, Y.~Yang, and M.~Winslett, ``Functional mechanism:
  Regression analysis under differential privacy,'' \emph{Proc. VLDB Endow.},
  vol.~5, no.~11, p. 1364–1375, 2012.

\bibitem{OLS_2}
A.~Nikolov, K.~Talwar, and L.~Zhang, ``The geometry of differential privacy:
  The sparse and approximate cases,'' in \emph{Proceedings of the Forty-Fifth
  Annual ACM Symposium on Theory of Computing}.\hskip 1em plus 0.5em minus
  0.4em\relax New York, NY, USA: Association for Computing Machinery, 2013, p.
  351–360.

\bibitem{OLS_3}
T.~Nguyen, X.~Xiao, Y.~Yang, S.~Hui, H.~Shin, and J.~Shin, ``Collecting and
  analyzing data from smart device users with local differential privacy,''
  \emph{CoRR}, 06 2016.

\bibitem{OLS_4}
O.~Sheffet, ``Differentially private ordinary least squares,'' in
  \emph{Proceedings of the 34th International Conference on Machine Learning -
  Volume 70}, 2017, p. 3105–3114.

\bibitem{OLS_5}
N.~Wang, X.~Xiao, Y.~Yang, J.~Zhao, S.~C. Hui, H.~Shin, J.~Shin, and G.~Yu,
  ``Collecting and analyzing multidimensional data with local differential
  privacy,'' in \emph{2019 IEEE 35th International Conference on Data
  Engineering (ICDE)}, 2019, pp. 638--649.

\bibitem{OLS_6}
D.~Alabi, A.~McMillan, J.~Sarathy, A.~Smith, and S.~Vadhan, ``Differentially
  private simple linear regression,'' \emph{Proceedings on Privacy Enhancing
  Technologies}, vol. 2022, pp. 184--204, 04 2022.

\bibitem{ridge}
A.~E. Hoerl and R.~W. Kennard, ``Ridge regression: Biased estimation for
  nonorthogonal problems,'' \emph{Technometrics}, vol.~12, no.~1, pp. 55--67,
  1970.

\bibitem{KPA_CPA}
S.~Li, C.~Li, G.~Chen, N.~G. Bourbakis, and K.-T. Lo, ``A general quantitative
  cryptanalysis of permutation-only multimedia ciphers against plaintext
  attacks,'' \emph{Signal Processing: Image Communication}, vol.~23, no.~3, pp.
  212--223, 2008.

\bibitem{multiplicative5}
K.~Liu, C.~Giannella, and H.~Kargupta, ``An attacker's view of distance
  preserving maps for privacy preserving data mining,'' in \emph{Proceedings of
  the 10th European Conference on Principles and Practice of Knowledge
  Discovery in Databases}, Berlin, Germany, September 2006, pp. 297--308.

\bibitem{poisoning}
M.~Jagielski, A.~Oprea, B.~Biggio, C.~Liu, C.~Nita-Rotaru, and B.~Li,
  ``Manipulating machine learning: Poisoning attacks and countermeasures for
  regression learning,'' in \emph{2018 IEEE Symposium on Security and Privacy
  (SP)}, 2018, pp. 19--35.

\bibitem{attack_1}
D.~Agrawal and C.~C. Aggarwal, ``On the design and quantification of privacy
  preserving data mining algorithms,'' in \emph{Proceedings of the Twentieth
  ACM SIGMOD-SIGACT-SIGART Symposium on Principles of Database Systems}, ser.
  PODS '01.\hskip 1em plus 0.5em minus 0.4em\relax New York, NY, USA:
  Association for Computing Machinery, 2001, p. 247–255.

\bibitem{noise_4}
H.~Kargupta, S.~Datta, Q.~Wang, and K.~Sivakumar, ``On the privacy preserving
  properties of random data perturbation techniques,'' in \emph{Third IEEE
  International Conference on Data Mining}, 2003, pp. 99--106.

\bibitem{attack_2}
Z.~Huang, W.~Du, and B.~Chen, ``Deriving private information from randomized
  data,'' in \emph{Proceedings of the 2005 ACM SIGMOD International Conference
  on Management of Data}, ser. SIGMOD '05.\hskip 1em plus 0.5em minus
  0.4em\relax New York, NY, USA: Association for Computing Machinery, 2005, p.
  37–48.

\bibitem{attack_3}
R.~H. Keshavan, A.~Montanari, and S.~Oh, ``Matrix completion from noisy
  entries,'' \emph{J. Mach. Learn. Res.}, vol.~11, p. 2057–2078, 2010.

\bibitem{DP_3}
J.~Blocki, A.~Blum, A.~Datta, and O.~Sheffet, ``The johnson-lindenstrauss
  transform itself preserves differential privacy,'' in \emph{2012 IEEE 53rd
  Annual Symposium on Foundations of Computer Science}, 2012, pp. 410--419.

\bibitem{DP_2}
C.~Xu, J.~Ren, Y.~Zhang, Z.~Qin, and K.~Ren, ``Dppro: Differentially private
  high-dimensional data release via random projection,'' \emph{IEEE
  Transactions on Information Forensics and Security}, vol.~12, no.~12, pp.
  3081--3093, 2017.

\bibitem{DP_1}
C.~Dwork, F.~McSherry, K.~Nissim, and A.~Smith, ``Calibrating noise to
  sensitivity in private data analysis,'' \emph{Journal of Privacy and
  Confidentiality}, vol.~7, no.~3, p. 17–51, May 2017.

\bibitem{UCI}
\BIBentryALTinterwordspacing
``{UCI} machine learning repository.'' [Online]. Available:
  \url{http://archive.ics.uci.edu/ml}
\BIBentrySTDinterwordspacing

\bibitem{diabetes}
B.~Strack, J.~Deshazo, C.~Gennings, J.~L. Olmo~Ortiz, S.~Ventura, K.~Cios, and
  J.~Clore, ``Impact of hba1c measurement on hospital readmission rates:
  Analysis of 70,000 clinical database patient records,'' \emph{BioMed research
  international}, vol. 2014, p. 781670, 04 2014.

\bibitem{logistic_DP}
M.~Kim, J.~Lee, L.~Ohno-Machado, and X.~Jiang, ``Secure and differentially
  private logistic regression for horizontally distributed data,'' \emph{IEEE
  Transactions on Information Forensics and Security}, vol.~15, pp. 695--710,
  2020.

\bibitem{logistic_2018}
J.~H. Cheon, D.~Kim, Y.~Kim, and Y.~Song, ``Ensemble method for
  privacy-preserving logistic regression based on homomorphic encryption,''
  \emph{IEEE Access}, vol.~6, pp. 46\,938--46\,948, 2018.

\bibitem{LR_1}
O.~Hellevik, ``Linear versus logistic regression when the dependent variable is
  a dichotomy,'' \emph{Quality \& Quantity: International Journal of
  Methodology}, vol.~43, no.~1, pp. 59--74, 2009.

\bibitem{LR_2}
C.~M. Norris, W.~A. Ghali, L.~D. Saunders, R.~Brant, D.~Galbraith, P.~Faris,
  and M.~L. Knudtson, ``Ordinal regression model and the linear regression
  model were superior to the logistic regression models,'' \emph{Journal of
  Clinical Epidemiology}, vol.~59, no.~5, pp. 448--456, 2006.

\bibitem{compression}
L.~Wen, K.~Zhou, S.~Yang, and L.~Li, ``Compression of smart meter big data: A
  survey,'' \emph{Renewable and Sustainable Energy Reviews}, vol.~91, pp.
  59--69, 2018.

\end{thebibliography}

\end{document}